\documentclass{amsart}
\usepackage[utf8]{inputenc}
\usepackage{amssymb,amsmath,amsfonts}
\usepackage{amsaddr}
\usepackage{geometry}
\usepackage{array,multicol,multirow}
\usepackage{graphicx}
\usepackage{enumerate,enumitem,setspace}
\usepackage{verbatim}
\usepackage{pdfsync}
\usepackage{gastex}
\usepackage[lowtilde]{url}
\usepackage{xcolor}
\usepackage{algorithm,algorithmicx,algpseudocode}
\DeclareSymbolFont{rsfscript}{OMS}{rsfs}{m}{n}
\DeclareSymbolFontAlphabet{\mathrsfs}{rsfscript}
\usepackage{anyfontsize}

\renewcommand{\O}{\mathcal{O}}
\DeclareMathOperator{\lspan}{span}
\newcommand{\scalar}{\boldsymbol{\cdot}}
\newtheorem{theorem}{Theorem}
\newtheorem{corollary}[theorem]{Corollary}

\newtheorem{proposition}[theorem]{Proposition}

\newtheorem{problem}{Problem}

\newtheorem{remark}[theorem]{Remark}

\begin{document}
\title{Preimage problems\\for deterministic finite automata}
\author {Mikhail V. Berlinkov}
\address{Institute of Natural Sciences and Mathematics,\\
Ural Federal University, Ekaterinburg, Russia}
\email{m.berlinkov@gmail.com}
\author{Robert Ferens}
\email{robert.ferens@cs.uni.wroc.pl}
\address{Institute of Computer Science,\\
University of Wroc{\l}aw, Wroc{\l}aw, Poland}
\author{Marek Szyku{\l}a}
\email{msz@cs.uni.wroc.pl}
\address{Institute of Computer Science,\\
University of Wroc{\l}aw, Wroc{\l}aw, Poland}

\begin{abstract}
Given a~subset of states $S$ of a~deterministic finite automaton and a~word $w$, the preimage is the subset of all states mapped to a~state in $S$ by the action of $w$.
We study three natural problems concerning words giving certain preimages.
The first problem is whether, for a~given subset, there exists a~word \emph{extending} the subset (giving a~larger preimage).
The second problem is whether there exists a~\emph{totally extending} word (giving the whole set of states as a~preimage)---equivalently, whether there exists an~\emph{avoiding} word for the complementary subset.
The third problem is whether there exists a~\emph{resizing} word.
We also consider variants where the length of the word is upper bounded, where the size of the given subset is restricted, and where the automaton is strongly connected, synchronizing, or binary.
We conclude with a summary of the complexities in all combinations of the cases.

\smallskip
\noindent\textsc{Keywords}: avoiding word, extending word, extensible subset, reset word, synchronizing automaton
\end{abstract}
\maketitle
%%%%%%%%%%%%%%%%%%%%%%%%%%%%%%%%%%%%%%%%%%%%%%%%%%%%%%%%%%%%
\section{Introduction}

A~deterministic finite complete (semi)automaton $\mathrsfs{A}$ is a~triple $(Q,\Sigma,\delta)$, where $Q$ is the set of \emph{states}, $\Sigma$ is the input \emph{alphabet}, and $\delta\colon Q \times \Sigma \to Q$ is the \emph{transition function}.
We extend $\delta$ to a~function $Q \times \Sigma^* \to Q$ in the usual way.
Throughout the paper, by $n$ we always denote the number of states $|Q|$.

When the context is clear, given a~state $q \in Q$ and a~word $w \in \Sigma^*$, we write shortly $q\cdot w$ for $\delta(q,w)$.
Given a~subset $S \subseteq Q$, the \emph{image} of $S$ under the action of a~word $w \in \Sigma^*$ is $S\cdot w = \delta(S,w) = \{q\cdot w \mid q \in S\}$.
The \emph{preimage} is $S\cdot w^{-1} = \delta^{-1}(S,w) = \{q \in Q \mid q\cdot w \in S\}$.
If $S=\{q\}$, then we usually simply write $q\cdot w^{-1}$.

We say that a~word $w$ \emph{compresses} a~subset $S$ if $|S\cdot w| < |S|$, \emph{avoids} $S$ if $(Q\cdot w) \cap S = \emptyset$, \emph{extends} $S$ if $|S\cdot w^{-1}| > |S|$, and \emph{totally extends} $S$ if $S\cdot w^{-1} = Q$.
A~subset $S$ is \emph{compressible}, \emph{avoidable}, \emph{extensible}, and \emph{totally extensible}, if there is a~word that, respectively, compresses, avoids, extends and totally extends it.

\begin{remark}
A~word $w \in \Sigma^*$ is avoiding for $S \subseteq Q$ if and only if $w$ is totally extending for $Q \setminus S$.
\end{remark}

\begin{figure}[htb]\centering
\includegraphics{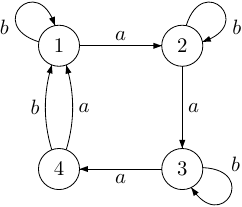}
\caption{The \v{C}ern\'{y} automaton with $4$ states.}\label{fig:cerny4}
\end{figure}

Fig.~\ref{fig:cerny4} shows an~example automaton.
For $S=\{2,3\}$, the shortest compressing word is $aab$, and we have $\{2,3\}\cdot aab = \{1\}$, while the shortest extending word is $ba$, and we have $\{2,3\} \cdot (ba)^{-1} = \{1,2\}\cdot b^{-1} = \{1,2,4\}$.

Note that the preimage of a~subset under the action of a~word can be smaller than the subset. In this case, we say that a~word \emph{shrinks} the subset (not to be confused with compressing when the image is considered).
For example, in Fig.~\ref{fig:cerny4}, subset $\{3,4\}$ is shrank by $b$ to subset $\{4\}$.

Note that shrinking a~subset is equivalent to extending its complement. Similarly, a~word totally extending a~subset also shrinks its complement to the empty set.
\begin{remark}
$|S\cdot w^{-1}| > |S|$ if and only if $|(Q \setminus S)\cdot w^{-1}| < |Q \setminus S|$, and
$S \cdot w^{-1} = Q$ if and only if $(Q \setminus S)\cdot w^{-1} = \emptyset$.
\end{remark}
Therefore, avoiding a~subset is equivalent to shrinking it to the empty set.

The \emph{rank} of a~word $w$ is the cardinality of the image $Q \cdot w$.
A~word of rank $1$ is called \emph{reset} or \emph{synchronizing}, and an~automaton that admits a~reset word is called \emph{synchronizing}.
Also, for a~subset $S \subseteq Q$, we say that a~word $w \in \Sigma^*$ such that $|S\cdot w|=1$ \emph{synchronizes} $S$.

Synchronizing automata serve as transparent and natural models of
various systems in many applications in different fields, e.g., in
coding theory \cite{BPR2010CodesAndAutomata,Jurgensen2008},
model testing of reactive systems \cite{Sandberg2005Survey},
robotics \cite{Na1986},
and biocomputing \cite{BAPLS2003}.
They also reveal interesting connections with many parts of mathematics.
For example, some of the recent works involve
group theory \cite{ArCaSt17},
representation theory \cite{AMSV2009RepresentationTheory},
computational complexity \cite{OM2010},
optimization and convex geometry \cite{GonzeJungers2016OnSynchronizingProbabilityFunction},
regular languages and universality \cite{RSX2012UniversalityProblems},
approximability \cite{GawrychowskiStraszak2015StrongInapproximability},
primitive sets of matrices \cite{BlondelJungersOlshevsky2015},
and graph theory \cite{GrechKisielewiczEdgeColoredDigraphs2016}.
For a~brief introduction to the theory of synchronizing automata we refer the reader to an excellent, though quite outdated, survey~\cite{Volkov2008Survey}.

The famous \v{C}ern\'{y} conjecture \cite{Cerny1964}, which was formally stated in~1969 during a~conference \cite{Volkov2008Survey}, is one of the most longstanding open problems in automata theory.
It states that a~synchronizing automaton has a~reset word of length at most $(n-1)^2$.
The currently best upper bound is cubic and has been improved recently \cite{Shitov2019} (cf.\ \cite{Szykula2018ImprovingTheUpperBound}).
Besides the conjecture, algorithmic issues are also important.
Unfortunately, the problem of finding a \emph{shortest} reset word is computationally hard \cite{Ep1990,OM2010}, and also its length approximation remains hard \cite{GawrychowskiStraszak2015StrongInapproximability}.
We also refer to surveys \cite{Sandberg2005Survey,Volkov2008Survey} dealing with algorithmic issues and the \v{C}ern\'{y} conjecture.

Compressing and extending a~subset in general play a~crucial role in the synchronization of automata and related areas.
In fact, all known algorithms finding a~reset word use finding words that either compresses or extends a~subset as subprocedures (e.g.~\cite{AG2016GreedyAlgorithms,BS2016AlgebraicSynchronizationCriterion,Ep1990,KKS2015ComputingTheShortestResetWords,RS2015ForwardAndBackward}).
Moreover, probably all proofs of upper bounds on the length of the shortest reset words use bounds on the length of words that compress (e.g.~\cite{AG2016GreedyAlgorithms,AV2005SynchronizingGeneralizedMonotonicAutomata,BS2016AlgebraicSynchronizationCriterion,BiskupPlandowski2009HuffmanCodes,Ep1990,GK2013AutomataRespectingIntervals,Szykula2018ImprovingTheUpperBound,Tr2007Aperiodic,Volkov2009ChainOfPartialOrders})
or extend (e.g.~\cite{BBP2011QuadraticUpperBoundInOneCluster,Berlinkov2013QuasiEulerianOneCluster,BS2016AlgebraicSynchronizationCriterion,Jungers2012SynchronizingProbabilityFunction,Kari2003Eulerian,Steinberg2011OneClusterPrime,Szykula2018ImprovingTheUpperBound}) some subsets.

In this paper, we study several problems about finding a word yielding a~certain preimage.
We provide a~systematic view of their computational complexity in various combinations of cases.

\subsection{Compressing a~subset}
The complexities of problems related to images of a~subset have been well studied.
It is known that given an~automaton $\mathrsfs{A}$ and a~subset $S \subseteq Q$, determining whether there is a~word that synchronizes it is PSPACE-complete \cite{Rystsov1983PolynomialCompleteProblems}.
The same holds even for strongly connected binary automata \cite{Vorel2014SubsetSynchronizationOfTransitiveAutomata}.

On the other hand, checking whether the automaton is synchronizing, i.e.\ whether there is a~word that synchronizes $Q$, can be solved in $\O(|\Sigma|n^2)$ time and space \cite{Cerny1964,Ep1990,Volkov2008Survey} and in $\O(n)$ average time and space when the automaton is randomly chosen ~\cite{Berlinkov2016OnTheProbabilityToBeSynchronizable}.
To this end, we verify whether all pairs of states are compressible.
Using the same algorithm, we can determine whether a~given subset is compressible.

Deciding whether there exists a~synchronizing word of a~given length is NP-complete \cite{Ep1990} (cf.~\cite{OM2010} for the complexity of the corresponding functional problems), even if the given automaton is binary.
The NP-completeness holds even when the automaton is Eulerian and binary \cite{Vorel2017ComplexityEulerian}, which immediately implies that for the class of strongly connected automata the complexity is the same.

However, deciding whether there exists a~word of a~given length that only compresses a~subset still can be solved in $\O(|\Sigma|n^2)$ time, as for every pair of states we can compute a~shortest word that compresses the pair.

The problems related to images have been also studied in other settings for both complexity and the bounds on the length of the shortest words, for example, in the case of a~nondeterministic automaton \cite{Rystsov1983PolynomialCompleteProblems}, in the case of a~partial deterministic finite automaton \cite{Martyugin2014ComputationalComplexity}, in the partial observability setting for various kinds of automata \cite{GLS2014SynchronizingStrategies}, and for the reachability of a~given subset in the case of a~deterministic finite automaton \cite{BondarVolkov2016CompletelyReachableAutomata,GonzeJungers2018OnCompletelyReachable}.

\subsection{Extending a~subset and our contributions}

\begin{table}[htb]\small\renewcommand{\arraystretch}{1.3}\centering
\newcommand{\htabsep}{1pt}
\newcommand{\rowt}[1]{\multirow{2}{*}{#1}}
\newcommand{\cold}[1]{\multicolumn{1}{c|}{#1}}
\newcommand{\colc}[1]{\multicolumn{1}{|c|}{#1}}
\newcommand{\colt}[1]{\multicolumn{2}{c|}{#1}}
\newcommand{\colttt}[1]{\multicolumn{4}{c|}{#1}}
\newcommand{\bm}{\boldmath}
\caption{The computational complexity of decision problems (new results are in bold): given an~automaton $\mathrsfs{A}=(Q,\Sigma,\delta)$ with $n$ states and a~subset $S \subseteq Q$, is there a~word $w \in \Sigma^*$ such that:}\label{tab:complexities_existence}
\begin{tabular}{|l|c|c|c|c|}\cline{2-5}
\cold{}                              &\colttt{Subclass of automata} \\\hline
\colc{\rowt{Problem}}                & All        & Strongly                    &\rowt{Synchronizing}                       & Str.\ con.   \\
                                     & automata   & connected                   &                                           & and synch.   \\[\htabsep]\hline
$|S\cdot w|=1$                       &\colt{PSPACE-c}                           &\rowt{$\O(1)$}                             &\rowt{$\O(1)$}\\
(reset word)                         &\colt{\cite{Rystsov1983PolynomialCompleteProblems,Vorel2014SubsetSynchronizationOfTransitiveAutomata}}&&\\[\htabsep]\hline
$|S\cdot w|<|S|$                     &\colt{$\O(|\Sigma|n^2)$}                  &\rowt{$\O(1)$}                             &\rowt{$\O(1)$}\\
(compressing word)                   &\colt{\cite{Cerny1964,Volkov2008Survey}}  &                                           &              \\[\htabsep]\hline

$|S\cdot w^{-1}|>|S|$                &\colt{\bf PSPACE-c}                       & \bf PSPACE-c                              &\rowt{$\O(1)$}\\
(Problem~\ref{pbm:extensible})       &\colt{(Thm.~\ref{thm:extensible})}        &(Prop.~\ref{pro:extensible_synchro})       &              \\[\htabsep]\hline
$S\cdot w^{-1}=Q$                    &\colt{\bf PSPACE-c}                       & \bm $\O(|\Sigma|n)$                       &\rowt{$\O(1)$}\\
(Problem~\ref{pbm:totallyextensible})&\colt{(Thm.~\ref{thm:extensible})}        & (Thm.~\ref{thm:totallyextensible_synchro})&              \\[\htabsep]\hline

$|S\cdot w^{-1}|>|S|$, $|S|\le k$    &\colt{\bm$\O(|\Sigma|n^k)$}               & \bm $\O(|\Sigma|n^k)$                     &\rowt{$\O(1)$}\\
(Problem~\ref{pbm:extensible_small}) &\colt{(Prop.~\ref{pro:extensible_small})} & (Prop.~\ref{pro:extensible_small})        &              \\[\htabsep]\hline
$S\cdot w^{-1}=Q$, $|S|\le k$        &\colt{\bm$\O(|\Sigma|n^k+n^3)$}         & \bm $\O(|\Sigma|n)$                       &\rowt{$\O(1)$}\\
(Problem~\ref{pbm:totallyextensible_small})&\colt{(Prop.~\ref{pro:totallyextensible_small})}&(Thm.~\ref{thm:totallyextensible_synchro})&   \\[\htabsep]\hline
$|S\cdot w^{-1}|>|S|$, $|S|\ge n-k$  & \bf PSPACE-c                &\rowt{Open} & \bf PSPACE-c                              &\rowt{$\O(1)$}\\
(Problem~\ref{pbm:extensible_large}, $k \ge 2$) & (Thm.~\ref{thm:extensible_large}) &   & (Thm.~\ref{thm:extensible_large}) &    \\[\htabsep]\hline
$S\cdot w^{-1}=Q$, $|S|\ge n-k$                 &\colt{\bm $\O(|\Sigma|n^k+n^3)$}       & \bm $\O(|\Sigma|n)$               &\rowt{$\O(1)$}\\
(Problem~\ref{pbm:totallyextensible_large}, $k \ge 2$)&\colt{(Thm.~\ref{thm:totallyextensible_large})}&(Thm.~\ref{thm:totallyextensible_synchro})&    \\[\htabsep]\hline
$S \cdot w^{-1}=Q$, $|S|= n-1$       &\colt{\bm $\O(|\Sigma|n^2)$}              & \rowt{\bm $\O(|\Sigma|)$}                 &\rowt{$\O(1)$}\\
(Problem~\ref{pbm:avoidable_state})  &\colt{(Thm.~\ref{thm:avoiding_characterization})} &                                   &              \\[\htabsep]\hline

$|S\cdot w^{-1}| \neq |S|$           &\colt{\bm $\O(|\Sigma|n^3)$}              &\rowt{$\O(1)$}                             &\rowt{$\O(1)$}\\
(Problem~\ref{pbm:resize})           &\colt{(Thm.~\ref{thm:resize_algorithm})}  &                                           &              \\[\htabsep]\hline
\end{tabular}
\end{table}

\begin{table}[htb]\small\renewcommand{\arraystretch}{1.3}\centering
\newcommand{\htabsep}{1pt}
\newcommand{\rowt}[1]{\multirow{2}{*}{#1}}
\newcommand{\cold}[1]{\multicolumn{1}{c|}{#1}}
\newcommand{\colc}[1]{\multicolumn{1}{|c|}{#1}}
\newcommand{\colt}[1]{\multicolumn{2}{c|}{#1}}
\newcommand{\colttt}[1]{\multicolumn{4}{c|}{#1}}
\newcommand{\bm}{\boldmath}
\caption{The computational complexity of decision problems (new results are in bold): given an~automaton $\mathrsfs{A}=(Q,\Sigma,\delta)$ with $n$ states, a~subset $S \subseteq Q$, and an~integer $\ell$ given in binary form, is there are a~word $w \in \Sigma^*$ of length $\le \ell$ such that:}\label{tab:complexities_length}
\begin{tabular}{|l|c|c|c|c|}\cline{2-5}

\cold{}                                                   & \colttt{Subclass of automata}                                                                                                                                                                                           \\ \hline
\colc{\rowt{Problem}}                                     & All        & Strongly                                                                                             &\rowt{Synchronizing}                              & Str.\ con.                                       \\
                                                          & automata   & connected                                                                                            &                                                  & and synch.                                       \\[\htabsep]\hline
$|S \cdot w|=1$                                           & \colt{PSPACE-c}                                                                                                   & NP-c                                             & NP-c                                             \\
(reset word)                                              & \colt{\cite{Rystsov1983PolynomialCompleteProblems,Vorel2014SubsetSynchronizationOfTransitiveAutomata}}            & \cite{Ep1990}                                    & \cite{Vorel2017ComplexityEulerian}               \\[\htabsep]\hline
$|S\cdot w|<|S|$                                          & \colt{$\O(|\Sigma|n^2)$}                                                                                          & $\O(|\Sigma|n^2)$                                & $\O(|\Sigma|n^2)$                                \\
(compressing word)                                        & \colt{\cite{Ep1990}}                                                                                              & \cite{Ep1990}                                    & \cite{Ep1990}                                    \\[\htabsep]\hline

$|S\cdot w^{-1}|>|S|$                                     & \colt{\bf PSPACE-c}                                                                                               & \bf PSPACE-c                                     & \bf NP-c                                         \\
(Problem~\ref{pbm:extensible_len})                        & \colt{(Subsec.~\ref{subsec:extensible_short})}                                                                    & (Subsec.~\ref{subsec:extensible_short})          & (Thm.~\ref{thm:avoiding_length_NP-c})            \\[\htabsep]\hline
$S\cdot w^{-1}=Q$                                         & \colt{\bf PSPACE-c}                                                                                               & \bf NP-c                                         & \bf NP-c                                         \\
(Problem~\ref{pbm:totallyextensible_len})                 & \colt{(Subsec.~\ref{subsec:extensible_short})}                                                                    & (Cor.~\ref{cor:totallyextending_large_len_sych}) & (Cor.~\ref{cor:totallyextending_large_len_sych}) \\[\htabsep]\hline

$|S\cdot w^{-1}|>|S|$, $|S|\le k$                         & \colt{\bm $\O(|\Sigma|n^k)$}                                                                                      & \bm $\O(|\Sigma|n^k)$                            & \bm $\O(|\Sigma|n^k)$                            \\
(Problem~\ref{pbm:extensible_small_len}                   & \colt{(Prop.~\ref{pro:extensible_small})}                                                                         & (Prop.~\ref{pro:extensible_small})               & (Prop.~\ref{pro:extensible_small})               \\[\htabsep]\hline

$S\cdot w^{-1}=Q$, $|S|\le k$                             & \colt{\bf NP-c}                                                                                                   & \bf NP-c                                         & \bf NP-c                                         \\
(Problem~\ref{pbm:totallyextensible_small_len})           & \colt{(Prop.~\ref{pro:totallyextensible_small_len})}                                                              & (Prop.~\ref{pro:totallyextensible_small_len})    & (Prop.~\ref{pro:totallyextensible_small_len})    \\[\htabsep]\hline

$|S\cdot w^{-1}|>|S|$, $|S|\ge n-k$                       & \bf PSPACE-c &\rowt{Open}                                                                                         & \bf PSPACE-c                                     & \bf NP-c                                         \\
(Problem~\ref{pbm:extensible_large_len}, $k\ge 2$)        & (Thm.~\ref{thm:extensible_large}) &                                                                               & (Thm.~\ref{thm:extensible_large})                & (Cor.~\ref{cor:totallyextending_large_len_sych}) \\[\htabsep]\hline

$S\cdot w^{-1}=Q$, $|S|\ge n-k$                           & \colt{\bf NP-c}                                                                                                   & \bf NP-c                                         & \bf NP-c                                         \\
(Problem~\ref{pbm:totallyextensible_large_len}, $k\ge 2$) & \colt{(Cor.~\ref{cor:totallyextending_large_len_sych})}                                                           & (Cor.~\ref{cor:totallyextending_large_len_sych}) & (Cor.~\ref{cor:totallyextending_large_len_sych}) \\[\htabsep]\hline

$S \cdot w^{-1}=Q$, $|S|= n-1$                            & \colt{\bf NP-c}                                                                                                   & \bf NP-c                                         & \bf NP-c                                         \\
(Problem~\ref{pbm:avoidable_state_len})                   & \colt{(Thm.~\ref{thm:avoiding_length_NP-c})}                                                                      & (Thm.~\ref{thm:avoiding_length_NP-c})            & (Thm.~\ref{thm:avoiding_length_NP-c})            \\[\htabsep]\hline
$|S\cdot w^{-1}| \neq |S|$                                & \colt{\bm $\O(|\Sigma|n^3)$}                                                                                      & \bm $\O(|\Sigma|n^3)$                            & \bm $\O(|\Sigma|n^3)$                            \\
(Problem~\ref{pbm:resize_len})                            & \colt{(Thm.~\ref{thm:resize_algorithm})}                                                                          & (Thm.~\ref{thm:resize_algorithm})                & (Thm.~\ref{thm:resize_algorithm})                \\[\htabsep]\hline

\end{tabular}
\end{table}

In contrast to the problems related to images (compression), the complexity of the problems related to preimages has not been thoroughly studied in the literature.
In the paper, we fill this gap and give a~comprehensive analysis of all basic cases.
We study three families of problems.
As noted before, extending is equivalent to shrinking the complementary subset, hence we need to deal only with the extending word problems.
Similarly, totally extending words are equivalent to avoiding the complement, thus we do not need to consider avoiding a set of states separately.

\noindent\textbf{Extending words}:
Our first family of problems is the question whether there exists an~extending word (Problems~\ref{pbm:extensible}, \ref{pbm:extensible_len}, \ref{pbm:extensible_small}, \ref{pbm:extensible_small_len}, \ref{pbm:extensible_large}, \ref{pbm:extensible_large_len} in this paper).

This is motivated by the fact that finding such a~word is the basic step of the so-called \emph{extension method} of finding a~reset word, which is used in many proofs and also some algorithms.
The extension method of finding a~reset word is as follows:
we start from some singleton $S_0 =\{q\}$ and iteratively find extending words $w_1,\ldots,w_k$ such that $|S_0 \cdot w_1^{-1} \cdots w_i^{-1}| > |S_0 \cdot w_1^{-1} \cdots w_{i-1}^{-1}|$ for $1 \le i \le k$, and where $S_0 \cdot w_1^{-1} \cdots w_k^{-1} = Q$.
For finding a~short reset word one needs to bound the lengths of the extending words.
For instance, in the case of synchronizing Eulerian automata, the fact that there always exists an~extending word of length at most $n-1$ implies the upper bound $(n-2)(n-1)+1$ on the length of the shortest reset words for this class \cite{Kari2003Eulerian} (the first extending step requires just one letter, as we can choose an arbitrary singleton).
In this case, a~polynomial algorithm for finding extending words has been proposed \cite{BS2016AlgebraicSynchronizationCriterion}.

\noindent\textbf{Totally extending words and avoiding}:
We study the problem whether there exists a~totally extending word (Problems~\ref{pbm:totallyextensible}, \ref{pbm:totallyextensible_len}, \ref{pbm:totallyextensible_small}, \ref{pbm:totallyextensible_small_len}, \ref{pbm:totallyextensible_large}, \ref{pbm:totallyextensible_large_len} in this paper).
The question of the existence of a~totally extending word is equivalent to the question of the existence of an~avoiding word for the complementary subset.

Totally extending words themselves can be viewed as a~generalization of reset words: a~word totally extending a~singleton to the whole set of states $Q$ is a~reset word.
If we are not interested in bringing the automaton into one particular state but want it to be in any of the states from a~specified subset, then it is exactly the question about totally extending word for our subset.
In view of applications of synchronization, this can be particularly useful when we deal with non-synchronizing automata, where reset words cannot be applied.

Avoiding word problem is a~recent concept that is dual to synchronization: instead of being in some states, we want not to be in them.
A~quadratic upper bound on the length of the shortest avoiding words of a~single state has been established \cite{Szykula2018ImprovingTheUpperBound}, which led to an improvement of the best known upper bound on the length of the shortest reset words (see also \cite{Shitov2019} for a very recent improvement of that improvement of the upper bound).
Furthermore, better upper bounds on the length of the shortest avoiding words would lead to further improvements; in particular, a subquadratic upper bound implies the upper bound on the reset threshold equal to $7n^3/48+o(n^3)$ \cite{GJT2014ANoteOnARecentAttempt}.
There is a precise conjecture that the shortest avoiding words have length at most $2n-2$ \cite[Open~Problem~1]{Szykula2018ImprovingTheUpperBound}.
The computational complexity of the problems related to avoiding, both a~single state or a~subset, has not been established before.
We give a~special attention to the problem of avoiding one state and a~small subset of states (totally extending a~large subset), as since they seem to be most important in view of their applications (and as we show, the complexity grows with the size of the subset to avoid).

\noindent\textbf{Resizing}: Shrinking a~subset is dual to extending, i.e.\ shrinking a~subset means extending its complement. Therefore, the complexity immediately transfers from the previous results.
However, in Section~\ref{sec:resize} we consider the problem of determining whether there is a~word whose inverse action results in a~subset having a~different size, that is, either extends the subset or shrinks it (Problems~\ref{pbm:resize},~\ref{pbm:resize_len}).

Interestingly, in contrast with the computationally difficult problems of finding a~word that extends the subset and finding a~word that shrinks the subset, for this variant there exists a~polynomial algorithm finding a~shortest resizing word in all cases.

We can mention that in some cases extending and shrinking words are related, and it may be enough to find either one.
For instance, this is used in the so-called \emph{averaging trick}, which appears in several proofs \cite{BS2016AlgebraicSynchronizationCriterion,Jungers2012SynchronizingProbabilityFunction,Kari2003Eulerian,Steinberg2011AveragingTrick}.

\noindent\textbf{Summary}:
For all the problems we consider the subclasses of strongly connected, synchronizing, and binary automata.
Also, we consider the problems where an~upper bound on the length of the word is additionally given in a~binary form in the input.
Since, in most cases, the problems are computationally hard, in Section~\ref{sec:extending_small} and Section~\ref{sec:extending_large}, we consider the complexity parameterized by the size of the given subset.

Table~\ref{tab:complexities_existence} and Table~\ref{tab:complexities_length} summarize our results together with known results about compressing words.
For the cases where a~polynomial algorithm exists, we put the time complexity of the best one known.
All the hardness results hold also in the case of a~binary alphabet.

%%%%%%%%%%%%%%%%%%%%%%%%%%%%%%%%%%%%%%%%%%%%%%%%%%%%%%%%%%%%
\section{Extending a~subset in general}\label{sec:extending}

\subsection{Unbounded word length}

In the first studied case, we do not have any restriction on the given subset $S$ neither on the length of the extending word.
We deal with the following problems:
\begin{problem}[Extensible subset]\label{pbm:extensible}
Given $\mathrsfs{A}=(Q,\Sigma,\delta)$ and a~subset $S \subseteq Q$, is $S$ extensible?
\end{problem}
\begin{problem}[Totally extensible subset]\label{pbm:totallyextensible}
Given $\mathrsfs{A}=(Q,\Sigma,\delta)$ and a~subset $S \subseteq Q$, is $S$ totally extensible?
\end{problem}

\begin{theorem}\label{thm:extensible}
Problem~\ref{pbm:extensible} and Problem~\ref{pbm:totallyextensible} are PSPACE-complete, even if $\mathrsfs{A}$ is strongly connected.
\end{theorem}
\begin{proof}
To solve one of the problems in NPSPACE, we guess the length of a word $w$ with the required property, and then guess the letters of $w$ from the end.
Of course, we do not store $w$, which may have exponential length, but just keep the subset $S\cdot u^{-1}$, where $u$ is the current suffix of $w$.
The current subset can be stored in $\O(n)$, and since there are $2^n$ different subsets, $|w| \le 2^n$ and the current length also can be stored in $\O(n)$. By Savitch's theorem, the problems are in PSPACE.

For PSPACE-hardness, we construct a reduction from the problem of determining whether an intersection of regular languages given as DFAs is non-empty.
We create one instance for both problems that consists of a strongly connected automaton and a subset $S$ extensible if and only if it is also totally extensible, which is simultaneously equivalent to the non-emptiness of the intersection of the given regular languages.

Let $(\mathcal{D}_i)_{i \in \{1,\ldots,m\}}$ be the given sequence of DFAs with an $i$-th automaton $\mathcal{D}_i=(Q_i,\Sigma,\delta_i,s_i,F_i)$ recognizing a language $L_i$, where $Q_i$ is the set of states, $\Sigma$ is the common alphabet, $\delta_i$ is the transition function, $s_i$ is the initial state, and $F_i$ is the set of final states.
The problem whether there exists a word accepted by all $\mathcal{D}_1,\ldots,\mathcal{D}_m$ (i.e.\ the intersection of $L_i$ is non-empty) is a well known PSPACE-complete problem, called Finite Automata Intersection \cite{Kozen1977}.
We can assume that the DFAs are \emph{minimal}; in particular, they do not have unreachable states from the initial state, otherwise, we may easily remove them in polynomial time.

For each $\mathcal{D}_i$ we choose an arbitrary $f_i \in F_i$.
Let $M = \sum_{i=1}^m |Q_i|$.
We construct the (semi)automaton $\mathcal{D}'=(Q',\Sigma',\delta')$  and define $S \subseteq Q'$ as an instance of our both problems.
The scheme of the automaton is shown in Fig.~\ref{fig:extending_reduction}.

\begin{figure}[htb]\centering
\includegraphics{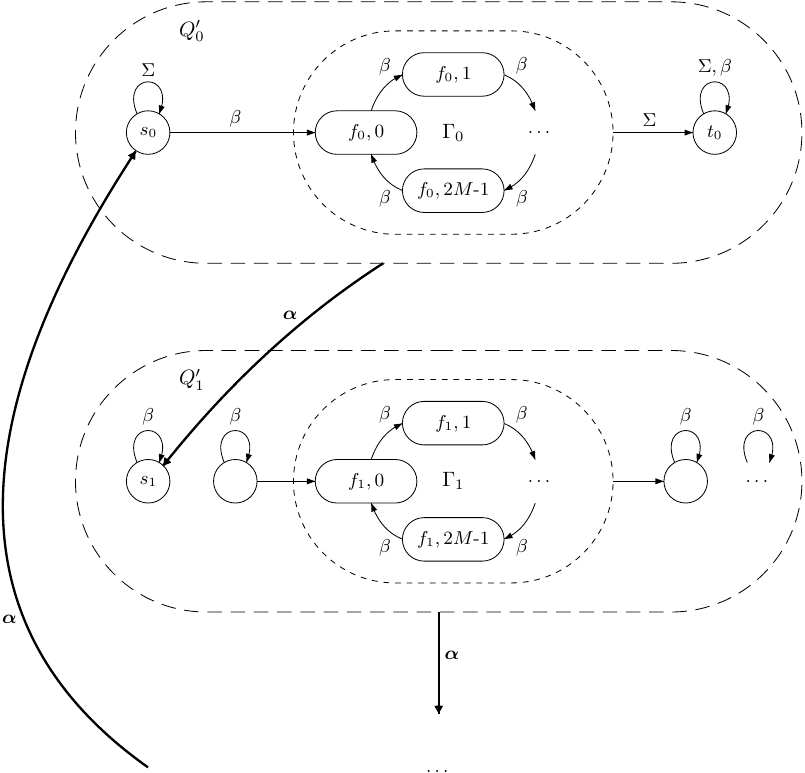}
\vspace{1cm}
\caption{The automaton $\mathcal{D}'$ from the proof of Theorem~\ref{thm:extensible}.} \label{fig:extending_reduction}
\end{figure}

\begin{itemize}
\item For $i \in \{0,1,\ldots,m\}$, let $\Gamma_i = \{f_i\} \times \{0,\ldots,2M-1\}$ be fresh states and let $Q_i' = (Q_i \setminus \{f_i\}) \cup \Gamma_i$.
Let $Q_0'=\{s_0,t_0\} \cup \Gamma_0$, where $s_0$ and $t_0$ are fresh states.
Then $Q' = \bigcup_{i=0}^m Q_i'$.
\item $\Sigma' = \Sigma \cup \{\alpha,\beta\}$, where $\alpha$ and $\beta$ are fresh letters. 
\item $\delta'$ is defined by:
\begin{itemize}
\item For $q \in Q_i \setminus \{f_i\}$ and $a \in \Sigma$, we have
$$\delta'(q,a)=\begin{cases}
\delta_i(q,a) & \text{if } \delta_i(q,a) \neq f_i,\\
(f_i,0) & \text{otherwise}.
\end{cases}$$
\item For $a \in \Sigma$, we have
$$\delta'(t_0,a)=t_0,\quad \delta'(s_0,a)=s_0.$$
\item For $k \in \{0,\ldots,2M-1\}$, $i \in \{1,\ldots,m\}$, and $a \in \Sigma$, we have
\begin{align*}
& \delta'((f_0,k),a) =\ t_0,\\
& \delta'((f_i,k),a) = \begin{cases}
\delta_i(f_i,a) & \text{if } \delta_i( f_i,a) \neq f_i,\\
(f_i,0) & \text{otherwise}.
\end{cases}
\end{align*}
\item For $q \in Q'_i$, we have
$$\delta'(q,\alpha)=s_{(i+1)\bmod(m+1)}.$$
\item For $i \in \{0,\ldots,m\}$ and $k \in \{0,\ldots,2M-1\}$, we have
$$\delta'((f_i,k),\beta) = (f_i,k+1 \bmod 2M).$$
\item We have
$$\delta'(s_0,\beta) = (f_0,0).$$
\item For the remaining states $q \in Q' \setminus (\bigcup_{i=0}^{m} \Gamma_i \cup \{s_0\})$, we have
$$\delta'(q,\beta)=q.$$
\end{itemize}
\item The subset $S \subseteq Q'$ is defined as
$$S=\big(\bigcup\limits_{i=1}^{m} F_{i} \cap Q'\big) \cup \bigcup\limits_{i=0}^{m} \Gamma_i \cup \{s_0\}.$$
\end{itemize}

It is easy to observe that $\mathcal{D}'$ is strongly connected.
Take any $i,j \in \{0,\ldots,m\}$.
We show how to reach any state $q \in Q_j'$ from a state $p \in Q_i'$.
First, we can reach $s_j$ by $\alpha^{(m+1+j-i) \bmod (m+1)}$.
For $j \ge 1$, each state $q \in Q_j' \setminus \big(\Gamma_j \setminus \{(f_j,0)\}\big)$ is reachable from $s_j$, since $\delta'$ restricted to $\Sigma$ acts on $Q_j'$ as $\delta_j$ on $Q_j$ (with $f_j$ replaced by $(f_j,0)$) and $\mathcal{D}_j$ is minimal. 
For $j=0$, states $(f_0,0)$ and $t_0$ are reachable from $s_0$ by the transformations of $\beta$ and $\beta a$ respectively, for any $a\in \Sigma$.
States $q \in \Gamma_j$ can be reached from $(f_j,0)$ using $\delta_\beta$.

We will show the following statements:
\begin{enumerate}
\item[(1)] If $S$ is extensible in $\mathcal{D}'$, then the intersection of the languages $L_i$ is non-empty.
\item[(2)] If the intersection of the languages $L_i$ is non-empty, then $S$ is extensible to $Q'$ in $\mathcal{D}'$.
\end{enumerate}
This will prove that the intersection of the languages $L_i$ is non-empty if and only if $S$ is extensible, which is also equivalent to that $S$ is extensible to  $Q'$.

(1):
Observe that, for each $i \in \{0,\ldots,m\}$, if $(S \cdot w^{-1}) \cap \Gamma_i \neq \emptyset$, then $(S \cdot w^{-1}) \cap \Gamma_i = \Gamma_i$.
This follows by induction: the empty word possesses this property; the transformation $\delta_a$ of $a \in \Sigma \setminus \{\beta\}$ maps every state from $\Gamma_i$ to the same state, so it preserves the property; $\delta_\beta$ acts cyclically on $\Gamma_i$ so also preserves the property.

Suppose that $S$ is extensible by a word $w$.
Notice that, $M$ is an upper bound on the number of states in $Q' \setminus \bigcup_{i=0}^m \Gamma_i$ (for $m \ge 2$).
We also have $|S| \ge 1+(m+1)\cdot 2M$.
We conclude that $\Gamma_i \subseteq S \cdot w^{-1}$ for each $i \in \{0,\ldots,m\}$, since
$$|Q' \setminus \Gamma_i| \le m \cdot 2M + M \le (m+1) \cdot 2M < |S|,$$
so $(S \cdot w^{-1}) \cap \Gamma_i \neq \emptyset$ and then our previous observation $\Gamma_i \subseteq S \cdot w^{-1}$.

Now, the extending word $w$ must contain the letter $\alpha$.
For a contradiction, if $w \in (\Sigma' \setminus \{\alpha\})^*$, then if it contains a letter $a \in \Sigma$, then $S\cdot w^{-1}$ does not contain any state from $\Gamma_0 \cup \{t_0\}$, as the only outgoing edges from this subset are labeled by $\alpha$, $t_0 \notin S$, $\Gamma_0 \cdot \beta^{-1} = \Gamma_0$, and $\Gamma_0 \cdot a^{-1} = \emptyset$.
This contradicts the previous paragraph.
Also, $w$ cannot be of the form $\beta^k$, for $k \in \mathbb{N}$, since $S \cdot \beta^k = S$.
Hence, $w = w_p \alpha w_s$, where $w_p \in (\Sigma')^*$ and $w_s \in (\Sigma' \setminus \{\alpha\})$. 

Note that if $T$ is a subset of $Q'$ such that $T \cap Q'_i = \emptyset$ for some $i$, then also $(T\cdot u^{-1}) \cap Q'_{i'} = \emptyset$ for every word $u$ and some $i'$; because only $\alpha$ maps states $Q_i$ outside $Q_i$, and it acts cyclically on these sets. 
Hence, in this case, every preimage of $T$ does not contain some $\Gamma_{i'}$ set.
So $\{s_i \mid i \in \{0,\cdots,m\}\} \subseteq S \cdot (w_s)^{-1}$, since in the opposite case $\big(S \cdot (\alpha w_s)^{-1} \big) \cap Q'_i = \emptyset$ for some $i$.

Let $w'_s$ be the word obtained by removing all $\beta$ letters from $w_s$.
Note that, for every $i \in \{1,\ldots,m\}$ and every suffix $u$ of $w_s$, we have $(S \cdot u^{-1}) \cap Q'_i = (S \cdot (\beta u)^{-1}) \cap Q'_i$.
Hence, $(S \cdot w_s^{-1}) \cap (Q' \setminus Q'_0) = S \cdot (w'_s)^{-1} \cap (Q' \setminus Q'_0)$.

Now, the word $w'_s$ is in $\Sigma^*$, and $S \cdot w_s^{-1}$ contains $s_i$ for all $i \in \{1,\ldots,m\}$.
Hence, the action of $w'_s$ maps $s_i$ to either a state in $F_i \setminus \{f_i\}$ or $(f_i,0)$, which means that $w'_s$ maps $s_i$ to $F_i$ in $\mathcal{D}_i$.
Therefore, $w'_s$ is in the intersection of the languages $L_i$.

(2):
Suppose that the intersection of the languages $L_i$ is non-empty, so there exists a word $w \in \Sigma^*$ such that $s_i \cdot w \in F_i$ for every $i$.
Then we have $S \cdot (\alpha w)^{-1} = Q'$, thus $S$ is extensible to $Q'$.
\end{proof}

We ensure that both problems remain PSPACE-complete in the case of a~binary alphabet, which follows from the following theorem.
\begin{theorem}\label{thm:extending_binary}
Given an~automaton $\mathrsfs{A}=(Q,\Sigma,\delta)$ and a~subset $S \subseteq Q$, we can construct in polynomial time a~binary automaton $\mathrsfs{A'}=(Q',\{a',b'\},\delta')$ and a~subset $S' \subseteq Q'$ such that:
\begin{enumerate}
\item[(1)] $\mathrsfs{A}$ is strongly connected if and only if $\mathrsfs{A}'$ is strongly connected;
\item[(2)] $S'$ is extensible in $\mathrsfs{A}'$ if and only if $S$ is extensible in $\mathrsfs{A}$;
\item[(3)] $S'$ is totally extensible in $\mathrsfs{A}'$ if and only if $S$ is totally extensible in $\mathrsfs{A}$.
\end{enumerate}
\end{theorem}
\begin{proof}
Let $\Sigma = \{a_0,\ldots,a_{k-1}\}$.
The idea is as follows:
We reduce $\mathrsfs{A}$ to a~binary automaton $\mathrsfs{A}'$ that consists of $k$ copies of $\mathrsfs{A}$.
The first letter $a$ acts in an $i$-th copy as the letter $a_i$ in $\mathrsfs{A}$.
The second letter $b$ acts cyclically on these copies.
Then we define $S'$ to contain states from $S$ in the first copy and all states from the other copies.
The construction is shown in Fig.~\ref{fig:extending_binarization}.

\begin{figure}[htb]\centering
\includegraphics{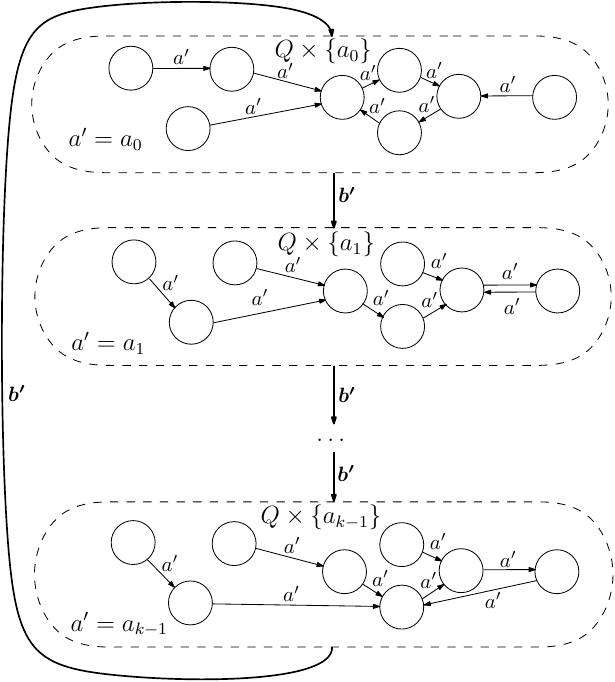}
\caption{The binary automaton $\mathcal{A}'$ from the proof of Theorem~\ref{thm:extending_binary}.}\label{fig:extending_binarization}
\end{figure}

We construct $\mathrsfs{A}'=(Q',\{a',b'\},\delta')$ with $Q' = Q \times \Sigma$ and $\delta'$ defined as follows:
$\delta'((q,a_i),a') = (\delta(q,a_i),a_i)$, and
$\delta'((q,a_i),b') = (q,a_{(i+1) \bmod k})$.
Clearly, $\mathrsfs{A}'$ can be constructed in $\O(nk)$ time, where $k=|\Sigma|$.

(1): Suppose that $\mathrsfs{A}$ is strongly connected; we will show that $\mathrsfs{A}'$ is also strongly connected.
Let $(q_1,a_i)$ and $(q_2,a_j)$ be any two states of $\mathrsfs{A}'$.
In $\mathrsfs{A}$, there is a~word $w$ such that $q_1\cdot w=q_2$.
Let $w'$ be the word obtained from $w$ by replacing every letter $a_h$ by the word $(b')^{h}a'(b')^{k-h}$.
Note that in $\mathrsfs{A}'$ we have
$$(p,a_0)\cdot (b')^{h}a'(b')^{k-h} = (p\cdot a_h,a_0),$$
hence $(q_1,a_0)\cdot w' = (q_1 \cdot w,a_0)$.
Then the action of the word $(b')^{k-i} w' (b')^{j}$ maps $(q_1,a_i)$ to $(q_2,a_j)$.

Conversely, suppose that $\mathrsfs{A}'$ is strongly connected, so every $(q_1,a_i)$ can be mapped to every $(q_2,a_j)$ by the action of a~word $w'$.
Then
$$w' = (b')^{h_1}a' \dots (b')^{h_{m-1}}a' (b')^{h_{m}},$$
for some $m \ge 1$ and $h_1,\ldots,h_m \ge 0$.
We construct $w$ of length $m-1$, where the $s$-th letter is $a_r$ with $r=(i+\Sigma_{j=1}^s h_{j}) \bmod k$. 
Then $w$ maps $q_1$ to $q_2$ in $\mathrsfs{A}$.

(2)~and~(3):
For $i \in \{0,\ldots,k-1\}$ we define $U_i = (Q \times \{\Sigma \setminus \{a_i\}\})$.
Observe that for any word $u' \in \{a',b'\}^*$, we have $U_i\cdot (u')^{-1} = U_j$ for some $j$, which depends on $i$ and the number of letters $b'$ in $u'$.

We define
$$S' = (S \times \{a_0\}) \cup U_0.$$
Suppose that $S$ is extensible in $\mathrsfs{A}$ by a~word $w$, and let $w'$ be the word obtained from $w$ as in~(1).
Then $(w')^{-1}$ maps $U_0$ to $U_0$, and $(S \times \{a_0\})$ to $(S \cdot w^{-1}) \times \{a_0\})$.
We have:
$$S' (w')^{-1} = ((S\cdot w^{-1}) \times \{a_0\}) \cup U_0,$$
and since $|S\cdot w^{-1}|>|S|$, this means that $w'$ extends $S'$.
By the same argument, if $w$ extends $S$ to $Q$, then $w'$ extends $S'$ to $Q'$.

Conversely, suppose that $S'$ is extensible in $\mathrsfs{A}'$ by a~word $w'$, and let $w$ be the word obtained from $w'$ as in~(1).
Then, for some $i$, we have
$$S'\cdot (w')^{-1} = ((S\cdot w^{-1}) \times \{a_i\}) \cup U_i,$$
and since $|U_0|=|U_i|$ it must be that $|S\cdot w^{-1}|>|S|$.
Also, if $S'\cdot (w')^{-1} = Q'$ then $S\cdot w^{-1}=Q$.
\end{proof}

Now, we consider the subclass of synchronizing automata.
We show that synchronizability does not change the complexity of the first problem, whereas the second problem becomes much easier.

\begin{proposition}\label{pro:extensible_synchro}
When the automaton is binary and synchronizing, Problem~\ref{pbm:extensible} remains PSPACE-complete.
\end{proposition}
\begin{proof}
From Theorem~\ref{thm:extensible}, Problem~\ref{pbm:extensible} is in PSPACE, as the algorithm works the same in the restricted case.

Problem~\ref{pbm:extensible} for binary and synchronizing automata is PSPACE-hard, as any general instance with a binary automaton $\mathrsfs{A}=(Q,\{a,b\},\delta)$ can be reduced to an equivalent instance with a binary synchronizing automaton $\mathrsfs{A}'$.
For this, we just add a~sink state $s$ and a~letter which synchronizes $Q$ to $s$.
Additionally, a~standard tree-like binarization is suitably used to obtain a~binary automaton $\mathrsfs{A}'$.

Formally, we construct a synchronizing binary automaton $\mathrsfs{A}'$ from the binary automaton $\mathrsfs{A}$ as follows.
We can assume that $Q = \{q_1,\ldots,q_n\}$.
Let $s$ be a~fresh state.
Let $Q' = Q \cup \{q^a_1,\ldots,q^a_n\}$.
We construct $\mathrsfs{A}'=(Q' \cup \{s\},\{a,b\},\delta')$, where $\delta'$ for all $i$ is defined as follows:
$\delta'(q_i,a)=q^a_i$, $\delta'(q_i,b)=s$, $\delta'(q^a_i,a)=\delta(q,a)$, and $\delta'(q^a_i,b) = \delta(q,b)$.
Then $bb$ is a synchronizing word for $\mathrsfs{A}'$, and each $S \subseteq Q$ is extensible in $\mathrsfs{A}'$ if and only if it is extensible in $\mathrsfs{A}$.
\end{proof}

\begin{theorem}\label{thm:totallyextensible_synchro}
When the automaton is synchronizing, Problem~\ref{pbm:totallyextensible} can be solved in $\O(|\Sigma|n)$ time and it is NL-complete.
\end{theorem}
\begin{proof}
Since $\mathrsfs{A}$ is synchronizing, Problem~\ref{pbm:totallyextensible} reduces to checking whether there is a~state $q \in S$ reachable from every state:
It is well known that a~synchronizing automaton has precisely one strongly connected \emph{sink} component that is reachable from every state.
If $w$ is a~reset word that synchronizes $Q$ to $p$, and $u$ is such that $p \cdot u = q$, then $wu$ extends $\{q\}$ to $Q$.
If $S$ does not contain a~state from the sink component, then every preimage of $S$ also does not contain these states.

The problem can be solved in $\O(|\Sigma|n)$ time, since the states of the sink component can be determined in linear time by Tarjan's algorithm \cite{Tarjan1972}.

It is also easy to see that the problem is in NL:
Guess a~state $q \in S$ and verify in logarithmic space that it is reachable from every state.

For NL-hardness, we reduce from ST-connectivity: Given a~graph $G=(V,E)$ and vertices $s,t$, check whether there is a~path from $s$ to $t$.
We will output a~synchronizing automaton $\mathrsfs{A}=(V,\Sigma,\delta)$ and $S \subseteq Q$ such that $S$ is extensible to $Q$ if and only if there is a~path from $s$ to $t$ in $G$.

First, we compute the maximum out-degree of $G$, and set $\Sigma = \Sigma' \cup \{\alpha\}$, where $|\Sigma'|$ is equal to the maximum out-degree.
We output $\mathrsfs{A}$ such that for every $q \in V$, every edge $(q,p) \in E$ is colored by a~different letter from $\Sigma'$.
If there is no outgoing edge from $q$, then we set the transitions of all letters from $\Sigma'$ to be loops.
If the out-degree is smaller than $|\Sigma'|$, then we simply repeat the transition of the last letter.
Next, we define $\delta(q,\alpha) = s$ for every $q \in V$.
Finally, let $S = \{t\}$.
The reduction uses logarithmic space since it requires only counting and enumerating through $V$ and $\Sigma'$.
The produced automaton $\mathrsfs{A}$ is synchronizing just by $\alpha$.

Suppose that there is a~path from $s$ to $t$.
Then there is a~word $w$ such that $\delta(s,w)=t$, and so $\{t\} \cdot (\alpha w)^{-1} = Q$.

Suppose that $\{t\}$ is extensible to $Q$ by some word $w$.
Let $w'$ be the longest suffix of $w$ that does not contain $\alpha$.
Since $\alpha^{-1}$ results in $\emptyset$ for any subset not containing $s$, it must be that $s \in \{t\} (w')^{-1}$.
Hence $\delta(s,w') = t$, and the path labeled by $w'$ is the path from $s$ to $t$ in $G$.
\end{proof}

Note that in the case of strongly connected synchronizing automaton, both problems have a~trivial solution, since every non-empty proper subset of $Q$ is totally extensible (by a~suitable reset word); thus they can be solved in constant time, assuming that we can check the size of the given subset and the number of states in constant time.

\subsection{Bounded word length}\label{subsec:extensible_short}

We turn our attention to the variants in which an~upper bound on the length of word $w$ is also given.

\begin{problem}[Extensible subset by short word]\label{pbm:extensible_len}
Given $\mathrsfs{A}=(Q,\Sigma,\delta)$, a~subset $S \subseteq Q$, and an~integer $\ell$ given in binary representation, is $S$ extensible by a~word of length at most $\ell$?
\end{problem}

\begin{problem}[Totally extensible subset by short word]\label{pbm:totallyextensible_len}
Given $\mathrsfs{A}=(Q,\Sigma,\delta)$, a~subset $S \subseteq Q$, and an~integer $\ell$ given in binary representation, is $S$ totally extensible by a~word of length at most $\ell$?
\end{problem}

Obviously, these problems remain PSPACE-complete (also when the automaton is strongly connected and binary), as we can set $\ell=2^n$, which bounds the number of different subsets of $Q$.
In this case, both the problems are reduced respectively to Problem~\ref{pbm:extensible} and Problem~\ref{pbm:totallyextensible}.

When the automaton is synchronizing, Problem~\ref{pbm:totallyextensible_len} is NP-complete, which will be shown in Corollary~\ref{cor:totallyextending_large_len_sych}.
Of course, Problem~\ref{pbm:extensible_len} remains PSPACE-complete for a~synchronizing automaton by the same argument as in the general case.

%%%%%%%%%%%%%%%%%%%%%%%%%%%%%%%%%%%%%%%%%%%%%%%%%%%%%%%%%%%%
\section{Extending small subsets}\label{sec:extending_small}

The complexity of the extending problems is caused by an~unbounded size of the given subset.
Note that in the proof of PSPACE-hardness in Theorem~\ref{thm:extensible} the used subsets and simultaneously their complements may grow with an~instance of the reduced problem, and it is known that the problem of the emptiness of intersection can be solved in polynomial time if the number of given DFAs is fixed.
Here, we study the computational complexity of the extending problems when the size of the subset is not larger than a~fixed $k$.

\subsection{Unbounded word length}

\begin{problem}[Extensible small subset]\label{pbm:extensible_small}
For a~fixed $k \in \mathbb{N} \setminus \{0\}$, given $\mathrsfs{A}=(Q,\Sigma,\delta)$ and a~subset $S \subseteq Q$ with $|S| \le k$, is $S$ extensible?
\end{problem}

\begin{proposition}\label{pro:extensible_small}
Problem~\ref{pbm:extensible_small} can be solved in $\O(|\Sigma|n^k)$ time.
\end{proposition}
\begin{proof}
We build the $k$-subsets automaton $\mathrsfs{A}^{\le k}=(Q^{\le k},\Sigma,\delta^{\le k}, S_0, F)$, where $Q^{\le k}=\{A \subseteq Q \colon |A| \le k\}$ and $\delta^{\le k}$ is naturally defined by the image of $\delta$ on a~subset.
Let the set of initial states be $I=\{A \in Q^{\le k} \colon |A \cdot a^{-1}| > |S| \text{ for some } a \in \Sigma\}$, and the set of final states be the set of all subsets of $S$.
A~final state can be reached from an~initial state if and only if $S$ is extensible in $\mathrsfs{A}$.
We can simply check this condition by a BFS algorithm.

Note that we can compute whether a~subset $A$ of size at most $k$ is in $I$ in $\O(|\Sigma|)$, by summing the sizes $|q\cdot a^{-1}|$ for all $q \in A$, where $|q\cdot a^{-1}|$ are computed during a~preprocessing, which takes $O(n)$ time for a~single $a \in \Sigma$.
Also, for a~given subset $A$ of size at most $k$, we can compute $T\cdot a$ in constant time (which depends only $k$).
Hence, the BFS works in linear time in the size of $\mathrsfs{A}^{\le k}$, so in $O(|\Sigma|n^k)$ time.
\end{proof}

\begin{problem}[Totally extensible small subset]\label{pbm:totallyextensible_small}
For a~fixed $k \in \mathbb{N} \setminus \{0\}$, given $\mathrsfs{A}=(Q,\Sigma,\delta)$ and a~subset $S \subseteq Q$ with $|S| \le k$, is $S$ totally extensible?
\end{problem}

For $k=1$, Problem~\ref{pbm:totallyextensible} is equivalent to checking if the automaton is synchronizing to the given state, thus can be solved in $\O(|\Sigma|n^2)$ time.
For larger $k$ we have the following:

\begin{proposition}\label{pro:totallyextensible_small}
Problem~\ref{pbm:totallyextensible_small} can be solved in $\O(|\Sigma|n^k + n^3)$ time.
\end{proposition}
\begin{proof}
Let $u$ be a~word of the minimal rank in $\mathrsfs{A}$.
We can find such a~word and compute the image $Q\cdot u$ in $\O(n^3+|\Sigma|n^2)$ time, using the well-known algorithm \cite[Algorithm~1]{Ep1990} generalized to non-synchronizing automata.
The algorithm just stops when there are no more compressible pairs of states contained in the current subset, and since the subset cannot be further compressed, the found word has the minimal rank.

For each $w \in \Sigma^*$ we have $S\cdot w^{-1}=Q$ if and only if $Q \cdot w \subseteq S$.
We can meet the required condition for $w$ if and only if $(Q\cdot u) \cdot w \subseteq S$.
Surely $|(Q\cdot u) \cdot w|=|Q\cdot u|$.
The desired word does not exist if the minimal rank is larger than $|S|=k$.
Otherwise, we can build the subset automaton $\mathrsfs{A}^{\le |Q \cdot u|}$ (similarly as in the proof of Proposition~\ref{pro:extensible_small}).
The initial subset is $Q \cdot u$.
If some subset of $S$ is reachable by a~word $w$, then the word $uw$ totally extends $S$ in $\mathrsfs{A}$.
Otherwise, $S$ is not totally extensible.
The reachability can be checked in at most $\O(|\Sigma| n^k)$ time.
However, if the rank $r$ of $u$ is less than $k$, the algorithm takes only $\O(|\Sigma|n^r)$ time.
\end{proof}

\subsection{Bounded word length}

We also have the two variants of the above problems when an~upper bound on the length of the word is additionally given.

\begin{problem}[Extensible small subset by short word]\label{pbm:extensible_small_len}
For a~fixed $k \in \mathbb{N} \setminus \{0\}$, given $\mathrsfs{A}=(Q,\Sigma,\delta)$, a~subset $S \subseteq Q$ with $|S| \le k$, and an~integer $\ell$ given in binary representation, is $S$ extensible by a word of length at most $\ell$?
\end{problem}

Problem~\ref{pbm:extensible_small_len} can be solved by the same algorithm in a Proposition~\ref{pro:extensible_small}, since the procedure can find a~shortest extending word.

\begin{problem}[Totally extensible small subset by short word]\label{pbm:totallyextensible_small_len}
For a~fixed $k \in \mathbb{N} \setminus \{0\}$, given $\mathrsfs{A}=(Q,\Sigma,\delta)$, a~subset $S \subseteq Q$ with $|S| \le k$, and an~integer $\ell$ given in binary representation, is $S$ totally extensible by a word of length at most $\ell$?
\end{problem}

\begin{proposition}\label{pro:totallyextensible_small_len}
For every $k$, Problem~\ref{pbm:totallyextensible_small_len} is NP-complete, even if the automaton is simultaneously strongly connected, synchronizing, and binary.
\end{proposition}
\begin{proof}
The problem is in NP, as the shortest extending words have length at most $\O(n^3+n^k)$ (since words of this length can be found by the procedure from Proposition~\ref{pro:totallyextensible_small}).

When we choose $S$ of size $1$, the problem is equivalent to finding a~reset word that maps every state to the state in $S$.
In~\cite{Vorel2017ComplexityEulerian} it has been shown that for Eulerian automata that are simultaneously strongly connected, synchronizing, and binary, deciding whether there is a~reset word of length at most $\ell$ is NP-complete. Moreover, in this construction, if there exists a~reset word of this length, then it maps every state to one particular state $s_2$ (see~\cite[Lemma~2.4]{Vorel2017ComplexityEulerian}).
Therefore, we can set $S=\{s_2\}$, and thus Problem~\ref{pbm:totallyextensible_small_len} is NP-complete.
\end{proof}

%%%%%%%%%%%%%%%%%%%%%%%%%%%%%%%%%%%%%%%%%%%%%%%%%%%%%%%%%%%%
\section{Extending large subsets}\label{sec:extending_large}

In this section, we consider the case where the subset $S$ contains all except at most a~fixed number of states $k$.

\subsection{Unbounded word length}

\begin{problem}[Extensible large subset]\label{pbm:extensible_large}
For a~fixed $k \in \mathbb{N} \setminus \{0\}$, given $\mathrsfs{A}=(Q,\Sigma,\delta)$ and a~subset $S \subseteq Q$ with $|Q\setminus S| \le k$, is $S$ extensible?
\end{problem}

\begin{problem}[Totally extensible large subset]\label{pbm:totallyextensible_large}
For a~fixed $k \in \mathbb{N} \setminus \{0\}$, given $\mathrsfs{A}=(Q,\Sigma,\delta)$ and a~subset $S \subseteq Q$ with $|Q\setminus S| \le k$, is $S$ totally extensible?
\end{problem}

Problem~\ref{pbm:totallyextensible_large} is equivalent to deciding the existence of an~avoiding word for a~subset $S$ of size $\le k$.
Note that Problem~\ref{pbm:extensible_large} and Problem~\ref{pbm:totallyextensible_large} are equivalent for $k=1$, when they become the problem of avoiding a~single given state.
Its properties will also turn out to be different than in the case of $k \ge 2$.
We give a~special attention to this problem, defined as follows, and study it separately.

\begin{problem}[Avoidable state]\label{pbm:avoidable_state}
Given $\mathrsfs{A}=(Q,\Sigma,\delta)$ and a~state $q \in Q$, is $\{q\}$ avoidable?
\end{problem}

The following result may be a~bit surprising, in view of that it is the only case where a general problem (i.e., Problems~\ref{pbm:extensible} and~\ref{pbm:totallyextensible}) remains equally hard when the subset size is additionally bounded.
We show that Problem~\ref{pbm:extensible_large} is PSPACE-complete for all $k \ge 2$, although the question about its complexity remains open for the class of strongly connected automata.

\begin{theorem}\label{thm:extensible_large}
Problem~\ref{pbm:extensible_large} is PSPACE-complete for every fixed $k \ge 2$, even if the given automaton is synchronizing and binary.
\end{theorem}
\begin{proof}
Problem~\ref{pbm:extensible_large} is in PSPACE as a special case of Problem~\ref{pbm:extensible}, which is PSPACE-complete (Thm.~\ref{thm:extensible}).

Now, we show a~reduction from Problem~\ref{pbm:totallyextensible}.
The idea is as follows.
We construct an automaton $\mathrsfs{A}'$ from the automaton $\mathrsfs{A}=(Q,\Sigma,\delta)$ given for Problem~\ref{pbm:totallyextensible}.
We add two new states, $e$ and $s$, and let the initial set $S'$ contain all the original states of $\mathrsfs{A}$.
State $s$ is a sink state ensuring that the automaton is synchronizing; it cannot be reached from $S'$ by inverse transitions.
Hence, to extend $S'$, one needs to get $e$, which is doable only by a new special letter $\alpha$.
This letter has the transition that shrinks all states $Q$ to the initial subset $S$ for the totally extensible problem.
This is done through an arbitrary selected state $f \in Q$.
Then we can reach $Q \cup \{e\}$ only by a totally extending word for $\mathrsfs{A}$.
The overall construction is presented in Fig.~\ref{fig:extending_large}.

\begin{figure}[htb]\label{fig:extending_large}\centering
\includegraphics{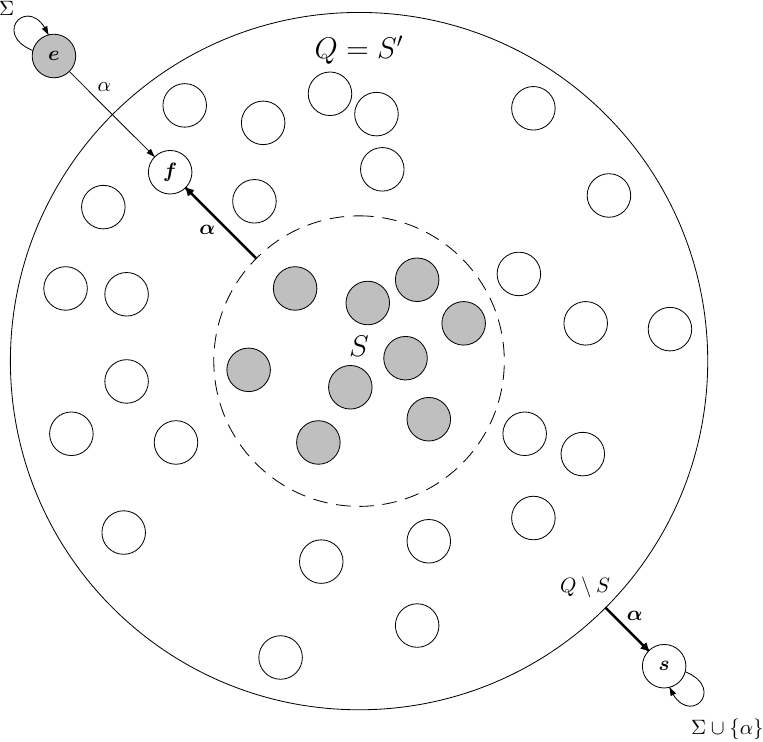}
\caption{The constructed automaton $\mathrsfs{A}'$:
States in $Q=S'$ have the transitions on $\Sigma$ as in $\mathrsfs{A}$.
The preimage of $S'=Q$ by $\alpha$ is marked by gray nodes and reflects the initial situation after applying for any subset containing $f$ and not containing $s$.}
\end{figure}

Let $\mathrsfs{A}=(Q,\Sigma,\delta)$ and $S\subseteq Q$ be an~instance of Problem~\ref{pbm:totallyextensible}.
We construct an~automaton $\mathrsfs{A}'=(Q'=Q \cup \{e,s\},\Sigma'=\Sigma \cup \{\alpha\},\delta')$, where $e,s$ are fresh states and $\alpha$ is a~fresh letter.
Let $f$ be an~arbitrary state from $Q \setminus S$ (if $S=Q$ then the problem is trivial).
We define $\delta'$ as follows:
\begin{enumerate}
\item $\delta'(q,a)=\delta(q,a)$ for $q \in Q$, $a \in \Sigma$;
\item $\delta'(q,a)=q$ for $q \in \{e,s\}$, $a \in \Sigma$;
\item $\delta'(q,\alpha)=f$ for $q \in S \cup \{e\}$;
\item $\delta'(q,\alpha)=s$ for $q \in (Q \cup \{s\}) \setminus S$.
\end{enumerate}
We define $S'=Q$.
Note that $|Q'\setminus S'|=2$, and hence automaton $\mathrsfs{A}'$ with $S'$ is an~instance of Problem~\ref{pbm:extensible_large} for $k=2$.
We will show that $S'$ is extensible in $\mathrsfs{A}'$ if and only if $S$ is totally extensible in $\mathrsfs{A}$.

If $S$ is totally extensible in $\mathrsfs{A}$ by a~word $w \in \Sigma^*$, we have $S' \cdot  (w\alpha)^{-1}=Q\setminus \{s\}$, which means that $S'$ is extensible in $\mathrsfs{A}'$.

Conversely, if $S'$ is extensible in $\mathrsfs{A}'$, then there is some extending word of the form $w\alpha$ for some $w \in \Sigma^*$, because $S'\cdot a^{-1}=S'$ for $a\in \Sigma$, $(Q' \setminus \{s\}) \cdot \alpha^{-1} \subseteq S'\cdot \alpha^{-1}$, and each reachable set (as a~preimage) is a~subset of $Q' \setminus \{s\}$.
We know that $S' \cdot (w\alpha)^{-1}=(S \cup \{e\}) \cdot w^{-1} = (S \cdot w^{-1}) \cup \{e\}$.
From the fact that $|S' \cdot (w\alpha)^{-1}| > |S'|$, we conclude that $S \cdot w^{-1} = Q$, so $S$ is totally extensible in $\mathrsfs{A}$.

Note that $\mathrsfs{A}'$ is synchronizing, since $Q' \cdot \alpha^2=\{f,s\} \cdot \alpha = \{s\}$.

Now, we show that we can reduce the alphabet to two letters.
Consider the application of the Theorem~\ref{thm:extending_binary} to Problem~\ref{pbm:extensible_large}. 
Note that the reduction in the proof keeps the size of complement set the same (i.e.\ $|Q' \setminus S'| = |Q'' \setminus S''|$, where $Q''$ and $S''$ are the set and the subset of states in the constructed binary automaton), so we can apply it.

Furthermore, we identify all the states of the form $(s,a)$ for $a \in \Sigma$ in the obtained binary automaton to one sink state $s''$.
In this way, we get a~synchronizing binary automaton (since $\mathrsfs{A}'$ is synchronizing).
The extending words remain the same, since the identified state $s''$ is not reversely reachable from $S''$, and $s''$ is not contained in the subset $S''$.

Finally, we conclude that the proof generalizes to the case of any $k \geq 2$ since we can add an arbitrary number of states with the same transitions as $e$.
\end{proof}

Now, we focus on totally extending words for large subsets, which we study in terms of avoiding small subsets.
First we provide a~complete characterization of single states that are avoidable:
\begin{theorem}\label{thm:avoiding_characterization}
Let $\mathrsfs{A}=(Q,\Sigma,\delta)$ be a~strongly connected automaton.
For every $q \in Q$, state $q$ is avoidable if and only if there exists $p \in Q \setminus \{q\}$ and $w \in \Sigma^*$ such that $q\cdot w=p\cdot w$.
\end{theorem}
\begin{proof}
First, for a given $q \in Q$, let $p \in Q \setminus \{q\}$ and $w \in \Sigma^*$ be such that $q\cdot w = p\cdot w$.
Since the automaton is strongly connected, there is a~word $w'$ such that $(p\cdot w)\cdot w'=(q\cdot w)\cdot w'=p$.
For each subset $S \subseteq Q$ such that $p \in S$ we have $p \in S\cdot ww'$. 
Moreover, if $q \in S$ then $|S\cdot ww'| < |S|$, because $\{q,p\}\cdot ww'=\{p\}$. 
If $q$ is not avoidable, then all subsets $Q\cdot (ww'),Q\cdot (ww')^2,\ldots$ contain $q$ and they form an~infinite sequence of subsets of decreasing cardinality, which is a~contradiction.

Now, consider the other direction.
Suppose for a~contradiction that a state $q \in Q$ is avoidable, but there is no state $p \in Q \setminus \{q\}$ such that $\{q,p\}$ can be compressed.
Let $u$ be a~word of the minimal rank in $\mathrsfs{A}$, and $v$ be a~word that avoids $q$. Then $w=uv$ has the same rank and also avoids $q$.
Let $\sim$ be the equivalence relation on $Q$ defined with a word $w$ as follows:
$$p_1\sim p_2 \iff p_1 \cdot w = p_2 \cdot w.$$
The equivalence class $[p]_\sim$ for $p \in Q$ is $(p\cdot w)\cdot w^{-1}$.
There are $|Q/{\sim}|=|Q\cdot w|$ equivalence classes and one of them is $\{q\}$, since $q$ does not belong to a~compressible pair of states.
For every state $p \in Q$, we know that $|(Q\cdot w) \cap [p]_\sim| \le 1$, because $[p]_\sim$ is compressed by $w$ to a~singleton and $Q\cdot w$ cannot be compressed by any word.
Note that every state $r \in Q\cdot w$ belongs to some class $[p]_\sim$.
From the equality $|Q/\sim|=|Q\cdot w|$ we conclude that for every class $[p]_\sim$ there is a~state $r \in (Q\cdot w) \cap [p]_\sim$, thus $|(Q\cdot w) \cap [p]_\sim| = 1$.
In particular, $1=|(Q\cdot w) \cap [q]_\sim|=|(Q\cdot w) \cap \{q\}|$.
This contradicts that $w$ avoids $q$.
\end{proof}

Note that if $\mathrsfs{A}$ is not strongly connected, then every state from a~strongly connected component that is not a~sink can be avoided.
If a~state belongs to a~sink component, then we can consider the sub-automaton of this sink component, and by Theorem~\ref{thm:avoiding_characterization} we know that given $q \in Q$, it is sufficient to check whether $q$ belongs to a~compressible pair of states.
Hence, Problem~\ref{pbm:avoidable_state} can be solved using the well-known algorithm (stage~1 in the proof of~\cite[Theorem~5]{Ep1990}) computing the pair automaton and performing a~breadth-first search with inverse edges on the pairs of states. It works in $\O(|\Sigma|n^2)$ time and $\O(n^2+|\Sigma|n)$ space.

We note that in a~synchronizing automaton all states are avoidable except a \emph{sink state}, which is a~state $q$ such that $q\cdot a=q$ for all $a \in \Sigma$. We can check this condition and hence verify if a~state is avoidable in a~synchronizing automaton in $\O(|\Sigma|)$ time.

The above algorithm does not find an~avoiding word but checks avoidability indirectly.
For larger subsets than singletons, we construct another algorithm finding a~word avoiding the subset, which also generalizes the idea from Theorem~\ref{thm:avoiding_characterization}.
From the following theorem, we obtain that Problem~\ref{pbm:totallyextensible_large} for a constant $k \ge 2$ can be solved in polynomial time. 

\begin{theorem}\label{thm:totallyextensible_large}
Let $\mathrsfs{A}=(Q,\Sigma,\delta)$, let $r$ be the minimum rank in $\mathrsfs{A}$ over all words, and let $S \subseteq Q$ be a~subset of size $\le k$.
We can find a~word $w$ such that $(Q\cdot w) \cap S = \emptyset$ or verify that it does not exist in $\O(|\Sigma|(n^{\min(r,k)}+n^2)+n^3)$ time and $\O(n^{\min(r,k)}+n^2+|\Sigma|n)$ space.
Moreover the length of $w$ is bounded by $\O(n^{\min(r,k)}+n^3))$.
\end{theorem}
\begin{proof}
Similarly to the proof of Theorem~\ref{thm:avoiding_characterization}, let $u$ be a~word of the minimal rank $r$ in $\mathrsfs{A}$ and let $\sim$ be the equivalence relation on $Q$ defined by word $u$ as follows:
$$p_1 \sim p_2 \iff p_1 \cdot u = p_2 \cdot u.$$
The equivalence class $[p]_\sim$ for $p \in Q$ is the set $(p\cdot u)\cdot u^{-1}$.
There are $|Q/{\sim}|=|Q\cdot u|$ equivalence classes.

First, we prove a~key observation that the image of each word starting with prefix $u$ has exactly one state in each equivalence class of $\sim$ relation.
Let $w=uw'$. Then the word $w$ has rank $r$ and its image is not compressible. For every state $p \in Q$, we know that $|(Q\cdot w) \cap [p]_\sim| \le 1$, because $[p]_\sim$ is compressed by $u$ to a~singleton and $Q\cdot w$ cannot be compressed by any word. Note that every state $q \in Q\cdot w$ belongs to some class $[p]_\sim$. From the equality $|Q/\sim|=|Q\cdot u|=|Q\cdot w|$ we conclude that for every class $[p]_\sim$ there is an~unique state $q_{[p]_\sim} \in (Q\cdot w) \cap [p]_\sim$. This proves the mentioned observation.

Now, we are going to show the following characterization: $S$ is avoidable if and only if there exist a~subset $Q' \subseteq Q\cdot u$ of size $|S/{\sim}|$ and a~word $w'$ such that $(Q' \cdot w') \cap ([s]_\sim \setminus S) \neq \emptyset$ for each $s \in S$. The idea of the characterization is illustrated in Fig.~\ref{fig:avoiding}.

\begin{figure}[htb]\centering
\includegraphics{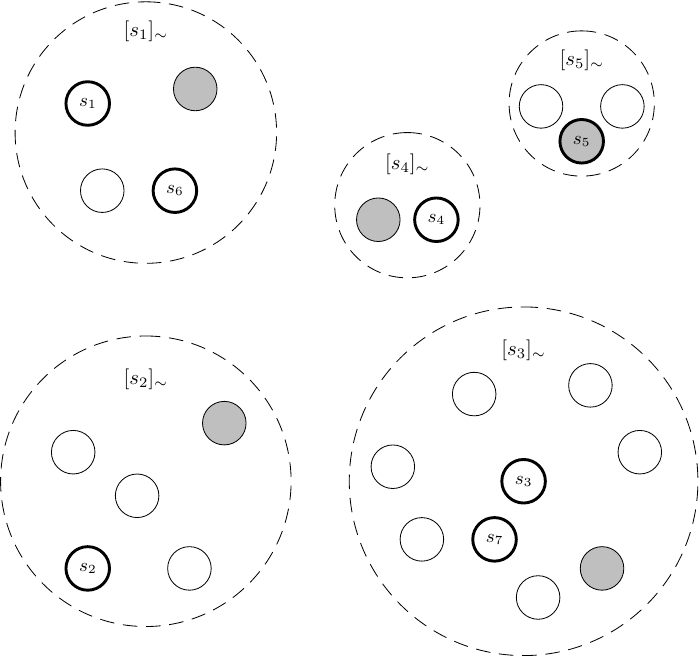}
\caption{The states of an~automaton divided by $\sim$. The states $s_i \in S$ are marked by bold border and the states $q_{[s_i]_\sim}$ in the image $Q\cdot uw'$ are filled. Every class has exactly one state in the image but can contain more than one state from $S$. If for each class this state is not in $S$, then $S$ is avoided. This is not the case in this example, because $s_5 \in Q\cdot uw'$.}\label{fig:avoiding}
\end{figure}

Suppose that $S$ is avoidable, and let $w'$ be an~avoiding word for $S$. Then the word $w=uw'$ also avoids $S$.
Observe that $Q \cdot w$ has an~unique state $q_{[p]_\sim} \in (Q\cdot w) \cap [p]_\sim$ for each class $[p]_\sim$.
Then for every state $s \in S$, we have $q_{[s]_\sim} \in [s]_\sim \setminus S$, because $w$ avoids $S$ and $q_{[s]_\sim} \in Q\cdot w$.
Notice that $[s]_\sim \cap S$ can contain more than one state, so the set $\{q_{[s]_\sim} \mid s \in S\}$ has size $|S/{\sim}|$, which is not always equal to $|S|$.
Therefore, there exists a~subset $Q' \subseteq Q\cdot u$ of size $|S/{\sim}|$ such that $Q'\cdot w' = \{q_{[s]_\sim} \mid s \in S\}$.
Now, we know that for every $s \in S$ we have $q_{[s]_\sim} \in Q' \cdot w'$ and $q_{[s]_\sim} \in [s]_\sim \setminus S$.
We conclude that, if $S$ is avoidable, then there exist a~subset $Q' \subseteq Q\cdot u$ of size $|S/{\sim}|$ and a~word $w'$ such that $(Q' \cdot w') \cap ([s]_\sim \setminus S) \neq \emptyset$ for every $s \in S$. 

Conversely, suppose that there is a~subset $Q' \subseteq Q\cdot u$ of size $|S/{\sim}|$ and a~word $w'$ such that $(Q' \cdot w') \cap ([s]_\sim \setminus S) \neq \emptyset$ for every $s \in S$.
Since in the image $Q \cdot uw'$ there is exactly one state in each equivalence class, we have
$((Q\cdot u) \setminus Q') \cdot w' \subseteq Q \setminus \bigcup_{s \in S}([s]_\sim) \subseteq Q \setminus S$, and by the assumption, $(Q' \cdot w') \cap S = \emptyset$.
Therefore, we get that $uw'$ is an~avoiding word for $S$.

This characterization gives us Alg.~\ref{alg:avoid_subset} to find $w$ or verify that $S$ cannot be avoided.

\begin{algorithm}[htb]
\caption{Avoiding a~subset.}\label{alg:avoid_subset}
\begin{algorithmic}[1]
\Require Automaton $\mathrsfs{A}(Q,\Sigma,\delta)$ and a~subset $S \subseteq Q$.
\State Find a~word $u$ of the minimal rank.
\State Compute $|S/{\sim}|$.
\ForAll{$Q' \subseteq Q \cdot u$ of size $|S/{\sim}|$}
  \If{there is a~word $w'$ such that $(Q' \cdot w') \cap ([s]_\sim \setminus S) \neq \emptyset$ for each $s \in S$}
    \State \Return $uw'$.
  \EndIf
\EndFor
\State \Return ``$S$ is unavoidable''.
\end{algorithmic}
\end{algorithm}

Alg.~\ref{alg:avoid_subset} first finds a~word $u$ of the minimal rank.
This can be done by in $\O(n^3+|\Sigma|n^2)$ time and $\O(n^2+|\Sigma|n)$ space by the well-known algorithm \cite[Algorithm~1]{Ep1990} generalized to non-synchronizing automata (cf. the proof of Proposition~\ref{pro:totallyextensible_small}.
For every subset $Q' \subseteq Q \cdot u$ of size $z=|S/{\sim}|$ the algorithm checks whether there is a~word $w'$ mapping $Q'$ to avoid $S$, but using its $\sim$-classes.
This can be done by constructing the automaton $\mathrsfs{A}^{z}(Q^{z},\Sigma,\delta^{z})$, where $\delta^{z}$ is $\delta$ naturally extended to $z$-tuples of states, and checking whether there is a~path from $Q'$ to a~subset containing a~state from each class $[s]_{\sim}$ but avoiding the states from $S$.
Note that since $Q'$ cannot be compressed, every reachable subset from $Q'$ has also size $|Q'|$.
The number of states in this automaton is $\binom{n}{z} \in \O(n^{z})$.
Also, note that we have to visit every $z$-tuple only once during a~run of the algorithm, and we can store it in $\O(n^{z}+|\Sigma|n)$ space.
Therefore, the algorithm works in $\O(n^3+|\Sigma|(n^2+n^{z}))$ time and $\O(n^2+n^{z}+|\Sigma|n)$ space. 

The length of $u$ is bounded by $\O(n^3)$, and the length of $w'$ is at most $\O(n^{z})$. Note that $z=|S/{\sim}| \le \min(r,|S|)$, where $r$ is the minimal rank in the automaton.
\end{proof}

\subsection{Bounded word length}

We now turn our attention to the variants of Problem~\ref{pbm:extensible_large}, Problem~\ref{pbm:totallyextensible_large}, and Problem~\ref{pbm:avoidable_state} where an~upper bound on the length of the word is additionally given.

\begin{problem}[Extensible large subset by short word]\label{pbm:extensible_large_len}
For a~fixed $k \in \mathbb{N} \setminus \{0\}$, given $\mathrsfs{A}=(Q,\Sigma,\delta)$, a~subset $S \subseteq Q$ with $|Q \setminus S| \le k$, and an~integer $\ell$ given in binary representation, is $S$ extensible by a word of length at most $\ell$?
\end{problem}

\begin{problem}[Totally extensible large subset by short word]\label{pbm:totallyextensible_large_len}
For a~fixed $k \in \mathbb{N} \setminus \{0\}$, given $\mathrsfs{A}=(Q,\Sigma,\delta)$, a~subset $S \subseteq Q$ with $|Q\setminus S| \le k$, and an~integer $\ell$ given in binary representation, is $S$ totally extensible by a word of length at most $\ell$?
\end{problem}

As before, both problems for $k=1$ are equivalent to the following:

\begin{problem}[Avoidable state by short word]\label{pbm:avoidable_state_len}
Given $\mathrsfs{A}=(Q,\Sigma,\delta)$, a~state $q \in Q$, and an~integer $\ell$ given in binary representation, is $\{q\}$ avoidable by a word of length at most $\ell$?
\end{problem}

Problem~\ref{pbm:extensible_large_len} for $k \ge 2$ obviously remains PSPACE-complete.
By the following theorem, we show that Problem~\ref{pbm:avoidable_state_len} is NP-complete, which then implies NP-completeness of Problem~\ref{pbm:totallyextensible_large_len} for every $k \ge 1$ (by Corollary~\ref{cor:totallyextending_large_len_sych}).

\begin{theorem}\label{thm:avoiding_length_NP-c}
Problem~\ref{pbm:avoidable_state_len} is NP-complete, even if the automaton is simultaneously strongly connected, synchronizing, and binary.
\end{theorem}
\begin{proof}
The problem is in NP, because we can non-deterministically guess a word $w$ as a certificate, and verify $q \notin Q\cdot w$ in $\O(|\Sigma|n)$ time.
If the state $q$ is avoidable, then the length of the shortest avoiding words is at most $\O(n^2)$ \cite{Szykula2018ImprovingTheUpperBound}.
Then we can guess an avoiding word $w$ of at most quadratic length and compute $Q\cdot w$ in $\O(n^3)$ time.

In order to prove NP-hardness, we present a polynomial-time reduction from the problem of determining the reset threshold in a specific subclass of automata, which is known to be NP-complete \cite[Theorem~8]{Ep1990}.
The reduction has two steps.
First, we construct a strongly connected synchronizing ternary automaton $\mathrsfs{A}'$ for which deciding about the length of an avoiding word is equivalent to determining the existence of a bounded length reset word in the original automaton.
Then, based on the ideas from~\cite{Berlinkov2014OnTwoAlgorithmicProblems}, we turn the automaton into a binary automaton $\mathrsfs{A}$, which still has the desired properties.

Let us have an instance of this problem from the Eppstein's proof of~\cite[Theorem~8]{Ep1990}.
Namely, for a given synchronizing automaton $\mathrsfs{B} = (Q_{\mathrsfs{B}}, \{\alpha_0,\alpha_1\}, \delta_{\mathrsfs{B}})$ and an integer $m>0$, we are to decide whether there is a reset word $w$ of length at most $m$.
We do not want to reproduce here the whole construction from the Eppstein proof but we need some ingredients of it.
Specifically, $\mathrsfs{B}$ is an automaton with a sink state $z \in Q_\mathrsfs{B}$, and there are two subsets $S=\{s_1,\ldots,s_d\}$ and $F \subseteq Q_{\mathrsfs{B}}$ with the following properties:
\begin{enumerate}
\item Each state $q \in Q_\mathrsfs{B} \setminus S$ is reachable from a state $s \in S$ through a (directed) path in the underlying digraph of $\mathrsfs{B}$.
\item For each state $s \in S$ and each word $w$ of length $m$, we have $\delta_{\mathrsfs{B}}(s,w) \in F \cup \{z\}$.
\item For each $f \in F$ we have $\delta_{\mathrsfs{B}}(f,\alpha_0) = \delta_{\mathrsfs{B}}(f,\alpha_1) = z$.
\item For each state $s \in S$ and a non-empty word $w \in \{\alpha_0,\alpha_1\}^{<m}$, we have $\delta_{\mathrsfs{B}}(s,w) \notin (F \cup S)$.
\end{enumerate}
In particular, it follows that each word of length $m+1$ is reset.
Deciding whether $\mathrsfs{B}$ has a reset word of length $m$ is NP-hard.

We transform the automaton $\mathrsfs{B}$ into $\mathrsfs{A}'$ as follows.
First, we add the subset $R = \{r_0,r_1, \ldots, r_{m}\}$ of states to provide that $z$ is not avoidable by words of length less than $m+1$.
The transitions of both letters are $\delta_{\mathrsfs{A}'}(r_i,\alpha_0) = \delta_{\mathrsfs{A}'}(r_i,\alpha_1) = r_{i+1}$ for $i=0,\ldots,m-1$, and $\delta_{\mathrsfs{A}'}(r_m,\alpha_0) = \delta_{\mathrsfs{A}'}(r_m,\alpha_1) = z$.

Secondly, we add a set of states $S' = \{ s'_1, \ldots, s'_{d}\}$ of size $d=|S|$ and a letter $\alpha_2$ to make the automaton strongly connected.
Letters $\alpha_0$ and $\alpha_1$ map $S'$ to the corresponding states from $S$, that is, $\delta_{\mathrsfs{A}'}(s'_{i},\alpha_0) = \delta_{\mathrsfs{A}'}(s'_{i},\alpha_1) = s_{i} \in S$.
Letter $\alpha_2$ connects states $r_0, s'_1,s'_2 \ldots, s'_{d}$ into one cycle, i.e.
$$\delta_{\mathrsfs{A}'}(r_0,\alpha_2) = s'_1,\quad \delta_{\mathrsfs{A}'}(s'_1,\alpha_2) = s'_2,\quad \ldots,\quad \delta_{\mathrsfs{A}'}(s'_{d-1},\alpha_2) = s'_{d}, \quad \delta_{\mathrsfs{A}'}(s'_{d},\alpha_2) = r_{0}.$$
We also set $\delta_{\mathrsfs{A}'}(s_{d},\alpha_2) = r_{1}$, $\delta_{\mathrsfs{A}'}(z,\alpha_2) = r_0$, and all the other transitions of $\alpha_2$ we define equal to the transitions of $\alpha_0$.

\begin{figure}[htb]\centering
\includegraphics{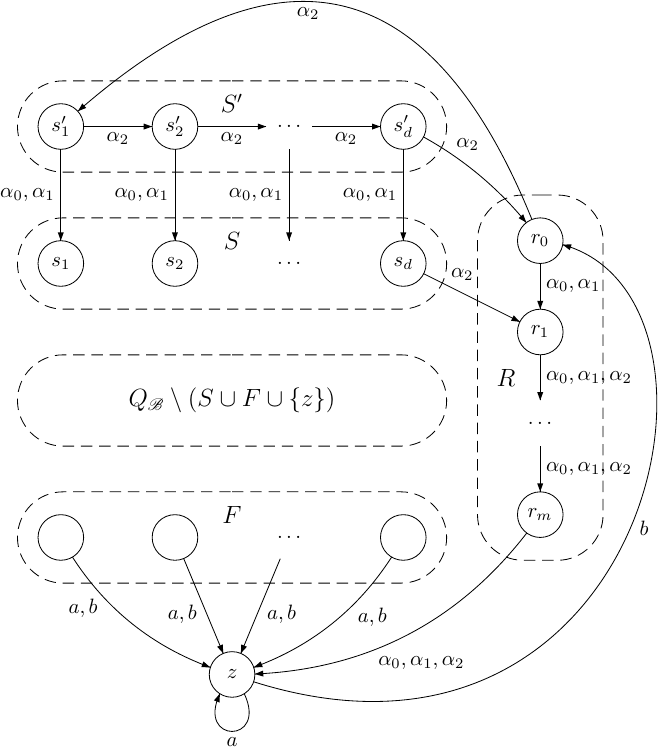}
\caption{The automaton $\mathrsfs{A}$ obtained from $\mathrsfs{A}'$ in the proof of Theorem~\ref{thm:avoiding_length_NP-c}. Here every state $q$ represents $\phi(q)$, and we have $\alpha_0\colon aa,ab$, $\alpha_1\colon ba$, and $\alpha_2\colon bb$.}\label{fig:avoiding_reduction}
\end{figure}

Finally, we transform $\mathrsfs{A}'$ to the final automaton $\mathrsfs{A}=(Q,\{a,b\},\delta)$.
We encode letters $\alpha_0,\alpha_1,\alpha_2$ by $2$-letter words over $\{a,b\}$ alike it was done in~\cite{Berlinkov2014OnTwoAlgorithmicProblems}.
Namely, for each state $q \in Q_{\mathrsfs{A}'} \setminus (F \cup \{z\})$, we add two new states $q^{a}, q^{b}$ and define their transitions as follows:
\begin{align*}
&\delta(q,a) = q^{a},\quad \delta(q^{a},a) = \delta(q^{a},b) = \delta_{\mathrsfs{A}'}(q,\alpha_0),\\
&\delta(q,b) = q^{b},\quad \delta(q^{b},a) = \delta_{\mathrsfs{A}'}(q,\alpha_1),\quad \delta(q^{b},b) = \delta_{\mathrsfs{A}'}(q,\alpha_2).
\end{align*}
Then, $aa,ab$ correspond to applying letter $\alpha_0$, $ba$ corresponds to applying letter $\alpha_1$, and $bb$ corresponds to applying letter $\alpha_2$.
Denote this encoding function by $\phi$, i.e. $\phi(\alpha_0) = aa$, $\phi(\alpha_1) = ba$, and $\phi(\alpha_2) = bb$.
We also extend $\phi$ to words over $\{\alpha_0,\alpha_1,\alpha_2\}^*$ as usual.
For simplicity, we denote also $\phi(q) = \{q, q^{a}, q^{b}\}$, and extend to subsets of $Q_\mathrsfs{A'}$ as usual.

It remains to define the transitions for $F \cup \{z\}$.
We set $\delta(z,a) = z$, $\delta(z,b) = r_0$, and $\delta(f,a) = \delta(f,b) = z$ for each $f \in F$.
Automaton $\mathrsfs{A}$ is shown in Fig.~\ref{fig:avoiding_reduction}.

Observe that $\mathrsfs{A}'$ is strongly connected:
$z$ is reachable from each state, from $z$ we can reach $r_0$ by $\alpha_2$, from $r_0$ we can reach every state from $S'$ by applying a power of letter $\alpha_2$, and we can reach every state of $S$ from the corresponding state from $S'$.
Then every state from $Q_\mathrsfs{B}$ is reachable from a state from $S$ by Property~1.
It follows that $\mathrsfs{A}$ is also strongly connected, since for every $q \in Q_\mathrsfs{A'}$, every state from $\phi(q)$ is reachable from $q$, and since for $F \cup \{z\}$ the outgoing edges correspond to those in $\mathcal{A}$.

Observe that $\mathrsfs{A}$ is synchronizing:
We claim that $a^{4m+6}$ is a reset word for $\mathrsfs{A}$.
Indeed, $aa$ does not map any state into $\phi(S')$.
Every word of length $m+1$ is reset for $\mathrsfs{B}$ and synchronizes to $z$, in particular, $\alpha_0^{m+1}$.
Since $\phi(\alpha_0^{m+1})=a^{2m+2}$ does not contain $bbb$, state $z$ cannot go to $S'$ by a factor of this word.
Hence, we have
$$\delta(Q,a^{2m+4}) \subseteq \{z\} \cup \phi(R).$$
Then, finally, $a^{2(m+1)}$ compresses $\{z\} \cup \phi(R)$ to $z$.

Now, we claim that the original problem of checking whether $\mathrsfs{B}$ has a reset word of length $m$ is equivalent to determining whether $z$ can be avoided in $\mathrsfs{A}$ by a word of length at most $2m+3$.

Suppose that $\mathrsfs{B}$ has a reset word $w$ of length $m$, and consider $u=\phi(\alpha_0w)b$.
Note that $\phi(\alpha_0)=aa$ does not map any state into $\phi(S')$ nor into $\phi(r_0)$.
Hence, we have
$$\delta(Q,\phi(\alpha_0)) \subseteq \phi(Q_\mathrsfs{B}) \cup \phi(R \setminus \{r_0\}).$$
Due to the definition of $\phi$, factor $bbb$ cannot appear in the image of words from $\{\alpha_0,\alpha_1\}^*$ by $\phi$.
Henceforth, $z$ cannot go to $S'$ by a factor of $\phi(w)$.
Since $|\phi(w)|=2m$ and to map $z$ into $\phi(r_m)$  we require a word of length $2m+1$, the factors of $\phi(w)$ do not map $z$ into $\phi(r_m)$.
Since also $w$ is a reset word for $\mathrsfs{B}$ that maps every state from $Q_\mathrsfs{B}$ to $z$, we have
$$\delta(\phi(Q_\mathrsfs{B}),\phi(w)) \subseteq \{z\} \cup \phi(R \setminus \{r_m\}).$$
By the definition of the transitions on $R \cup \{z\}$ (only $\phi(\alpha_2)$ maps $r_0$ outside), and since $|\phi(w)|=2m$, we also have
$$\delta(\phi(R \setminus \{r_0\}),\phi(w)) \subseteq \{z\} \cup \phi(R \setminus \{r_m\}).$$
Finally, we get that $\delta(\{z\} \cup \phi(R \setminus \{r_m\}),b) \subset R$, thus $u$ avoids $z$.

Now, we prove the opposite direction.
Suppose that state $z$ can be avoided by a word $u$ of length at most $2m+3$.
Then, by the definition of the transitions on $R$, $|u|=2m+3$ because $z \in \delta(R,w)$ for each $w$ of length at most $2(m+1)$.
Let $u = u' u'' u'''$ with $|u'|=2$, $|u''|=2m$, and $|u'''|=1$.

For words $w \in \{a,b\}^*$ of even length, we denote by $\tilde{\phi}^{-1}(w)$ the inverse image of encoding $\phi$ with respect to the definition on $\mathrsfs{A}'$, that is,
$\tilde{\phi}^{-1}(aa) = \tilde{\phi}^{-1}(ab) = \alpha_0$, $\tilde{\phi}^{-1}(ba) = \alpha_1$, $\tilde{\phi}^{-1}(bb) = \alpha_2$, which is extended to words of even length by concatenation.

First notice that $\tilde{\phi}^{-1}(u') \neq \alpha_2$. Otherwise $\{z,r_0,r_1,r_2,\dots,r_m\} \subseteq \delta(S' \cup R \cup \{z\}, \tilde{\phi}^{-1}(u'))$ whence by the definition of $R$ the word $u''u'''$ of length $2m+1$ cannot avoid $z$. Therefore $\tilde{\phi}^{-1}(u') \neq \alpha_2$ and $S \subseteq \delta(S \cup S',u')$.

If $\alpha_2$ is the second letter of $\tilde{\phi}^{-1}(u)$, then $s_d$ goes to $r_1$ and we get $\{r_1,r_2,\dots,r_m,z\}$ in the image of the prefix of $u$ of length $4$. Then, due to the definition of $R$, no word of length at most $2m$ can avoid $z$.

Hence, the first two letters of $\tilde{\phi}^{-1}(u)$ are either $\alpha_0$ or $\alpha_1$.

By Property~2 of $\mathrsfs{B}$, every zero-one word of length $m$ maps $s \in S$ into $\{z\} \cup F$. Since the letter $\alpha_2$ acts like $\alpha_0$ on $Q_\mathrsfs{B} \setminus S$ in $\mathrsfs{A}'$ and $\tilde{\phi}^{-1}(u'')$ starts with $\alpha_0$ or $\alpha_1$, $u''$ maps $S$ into $\{z\} \cup F$. If $u''$ maps some state to $F$, then by Property~3 $u$ cannot avoid $z$. Hence, $\tilde{\phi}^{-1}(u'')$ with all $\alpha_2$ replaced with $\alpha_0$ must be a reset word for $\mathrsfs{B}$.
\end{proof}

By a~corollary from Theorem~\ref{thm:avoiding_length_NP-c} and Theorem~\ref{thm:totallyextensible_large}, we complete our results about extending subsets.

\begin{corollary}\label{cor:totallyextending_large_len_sych}
Problem~\ref{pbm:totallyextensible_large_len} is NP-complete,
Problem~\ref{pbm:totallyextensible_len} is NP-complete when the automaton is synchronizing, and Problem~\ref{pbm:extensible_large_len} is NP-complete when the automaton is strongly connected and synchronizing.
They remain NP-complete when the automaton is simultaneously strongly connected, synchronizing, and binary.
\end{corollary}
\begin{proof}
NP-hardness for all the problems follows from Theorem~\ref{thm:avoiding_length_NP-c}, since we can set $S = Q \setminus \{q\}$.

Problem~\ref{pbm:totallyextensible_large_len} is solvable in NP as follows.
By Theorem~\ref{thm:totallyextensible_large} if there exists a~totally extending word, then there exists such a~word of polynomial length.
Thus we first run this algorithm, and if there is no totally extending word then we answer negatively.
Otherwise, we know that the length of the shortest totally extending words is polynomially bounded, so we can nondeterministically guess such a~word of length at most $\ell$ and verify whether it is totally extending.

Similarly, Problem~\ref{pbm:totallyextensible_len} is solvable in NP for synchronizing automata.
For a~synchronizing automaton there exists a~reset word $w$ of length at most $n^3$ \cite{Volkov2008Survey}.
Furthermore, if $S$ is totally extensible, then there must exist a~reset word $w$ such that $Q\cdot w = \{q\} \subseteq S$, which has length at most $n^3+n-1$.
Therefore, if the given $\ell$ is larger than this bound, we answer positively.
Otherwise, we nondeterministically guess a~word of length at most $\ell$ and verify whether it totally extends $S$.

By the same argument for Problem~\ref{pbm:extensible_large_len}, if the automaton is strongly connected and synchronizing, then for a~non-empty proper subset of $Q$ using a~reset word we can always find an~extending word of length at most $n^3+n-1$, thus the problem is solvable in NP. 
\end{proof}

%%%%%%%%%%%%%%%%%%%%%%%%%%%%%%%%%%%%%%%%%%%%%%%%%%%%%%%%%%%%
\section{Resizing a~subset}\label{sec:resize}

In this section we deal with the following two problems:

\begin{problem}[Resizable subset]\label{pbm:resize}
Given an~automaton $\mathrsfs{A}=(Q,\Sigma,\delta)$ and a~subset $S \subseteq Q$, is $S$ resizeable?
\end{problem}

\begin{problem}[Resizable subset by short word]\label{pbm:resize_len}
Given an~automaton $\mathrsfs{A}=(Q,\Sigma,\delta)$, a~subset $S \subseteq Q$, and an~integer $\ell$ given in binary representation, is $S$ resizeable by a word of length at most $\ell$?
\end{problem}

In contrast to the cases $|S\cdot w^{-1}|>|S|$ and $|S\cdot w^{-1}|<|S|$, there exists a~polynomial-time algorithm for both these problems.
Furthermore, we prove that if $S$ is resizeable, then the length of the shortest resizing words is at most $n-1$.

To obtain a polynomial-time algorithm, one could reduce Problem~\ref{pbm:resize} to the \emph{multiplicity equivalence of NFAs}, which is the problem whether two given NFAs have the same number of accepting paths for every word.
It can be solved in $\O(|\Sigma| n^4)$ time by a Tzeng's algorithm \cite{Tzeng1989}, assuming that arithmetic operations on real numbers have a unitary cost; this algorithm relies on linear algebra methods.
Alternatively, it can be solved in $\O(|\Sigma|^2 n^3)$ time by an algorithm of Archangelsky \cite{Archangelsky2003}.
It was noted by Diekert that the Tzeng's algorithm could be improved to $\O(|\Sigma| n^3)$ time \cite{Archangelsky2003} (unpublished).

However, to obtain the tight upper bound $n-1$ on the length we need to design and analyze a specialized algorithm for our problem. It is also based on the Tzeng's linear algebraic method.

\begin{theorem}\label{thm:resize_algorithm}
Assuming that in our computational model every arithmetic operation has a unitary cost, there is an algorithm with $\O(|\Sigma| n^3)$ time and $\O(|\Sigma|n+n^2)$ space complexity, which, given an~$n$-state automaton $\mathrsfs{A}=(Q,\Sigma,\delta)$ and a~subset $S \subseteq Q$, returns the minimum length $\ell$ such that $|S\cdot w^{-1}| \neq |S|$ for some word $w \in \Sigma^{\leq \ell}$ if it exists or reports that there is no such a word.
Furthermore, we always have $1 \le \ell \le n-1$.
\end{theorem}
\begin{proof}
The idea of the algorithm is based on the \emph{ascending chain condition}, often used for automata (e.g.\ \cite{Kari2003Eulerian,Pin1972Utilisation,Szykula2018ImprovingTheUpperBound}).
We need to introduce a few definitions from linear algebra.
We associate a natural linear structure with automaton $\mathrsfs{A}$.
By $\mathbb{R}^n$ we denote the real $n$-dimensional linear space of row vectors.
The value at an $i$-th entry of a vector $v \in \mathbb{R}^n$ we denote by $v(i)$. Without loss of generality, we assume that $Q=\{1,2,\dots,n\}$ and then assign to each subset $K\subseteq Q$ its \emph{characteristic vector} $[K] \in \mathbb{R}^n$, whose $i$-th entry $v(i)=1$ if $i \in K$, and $v(i)=0$, otherwise.
By $\lspan(S)$ we denote the linear span of $S \subseteq \mathbb{R}^n$.
The \emph{dimension} of a linear subspace $L$ is denoted by $\dim(L)$.

Each word $w \in \Sigma^*$ corresponds to a linear transformation of $\mathbb{R}^n$.
By $[w]$ we denote the matrix of this transformation in the standard basis $[1],\ldots,[n]$ of $\mathbb{R}^n$.
For example, if $\mathrsfs{A}$ is the automaton from Fig.~\ref{fig:cerny4}, then
$$[a]=\left(
  \begin{smallmatrix}
    0 & 1 & 0 & 0\\
    0 & 0 & 1 & 0\\
    0 & 0 & 0 & 1\\
    1 & 0 & 0 & 0
  \end{smallmatrix}
\right),\ [b]=\left(
  \begin{smallmatrix}
    1 & 0 & 0 & 0\\
    0 & 1 & 0 & 0\\
    0 & 0 & 1 & 0\\
    1 & 0 & 0 & 0
  \end{smallmatrix}
\right),\ [ba]=\left(
  \begin{smallmatrix}
    0 & 1 & 0 & 0\\
    0 & 0 & 1 & 0\\
    0 & 0 & 0 & 1\\
    0 & 1 & 0 & 0
  \end{smallmatrix}
\right).$$
Clearly, as the automaton is deterministic, the matrix $[w]$ has exactly one non-zero entry in each row.
In particular, $[w]$ is \emph{row stochastic}, which means that the sum of entries in each row is equal to $1$.
For every words $u,v \in \Sigma^*$, we have $[uv]=[u][v]$. 
By $[w]^T$ we denote the transpose of the matrix $[w]$.
The transpose corresponds to the preimage by the action of a word; one verifies that $[S \cdot w^{-1}]=[S][w]^T$.
For two vectors $v_1,v_2 \in \mathbb{R}^n$, we denote their usual inner (scalar) product by $v_1 \scalar v_2$.

\noindent\textit{Algorithm description}.
Now, we design the algorithm, which consists of two parts.

First, consider the auxiliary $\Call{Filter}{}$ function shown in Algorithm~\ref{alg:filter}.
Its goal is to filter a stream of vectors $g \in \mathbb{R}^n$, keeping only a subset of those vectors that are linearly independent.
To perform this subroutine efficiently, we maintain a sequence of vectors $G$ (basis) and a sequence of indices $I$, which are empty at the beginning.
Every time, we use the Gaussian approach to reduce the matrix of vectors from $G$ to a \emph{pseudo-triangular} form.
The sequence of (column) indices $I = (i_1,i_2,\ldots,i_k)$ and vectors $G = (g_1,\ldots,g_k)$ have the property that for each $j$, $1 \le j \le k$, there is exactly one vector from $\{g_1,\ldots,g_k\}$ with non-zero $i_j$-th entry, which contains $1$.

\begin{algorithm}[htb]
\caption{Filter.}\label{alg:filter}
\begin{algorithmic}[1]
\State $G \gets (), I \gets ()$.
\Comment{Global initialization}
\Function{Feed}{$g \in \mathbb{R}^n$} 
    \State $g' 
    \gets g - \sum_{r=1}^{k} {g(i_r)}\cdot g_r$
    \If{$g' = 0$}
        \State \Return \textbf{False}
    \Else
        \State $i' \gets \min(i \mid g(i) \neq 0)$
        \State $g' \gets g' / g'(i')$
        \ForAll{$g_r$ from $G$}
            \State $g_r \gets g_r - g_r(i') \cdot g'$
        \EndFor
        \State Append $g'$ to $G$
        \State Append $i'$ to $I$
        \State \Return \textbf{True}
    \EndIf
\EndFunction
\end{algorithmic}
\end{algorithm}

We begin with the first non-zero vector $g_1$ and put its smallest index $i$ of a non-zero entry to $I$, and the vector itself is normalized to have $1$ in the $i$-th entry.
Now, suppose we are given a vector $g$ and we have already built $G = (g_1,\ldots,g_k)$ and $I = (i_1,i_2,\dots,i_{k})$ with aforementioned properties.
Then, we just compute $g' = g - \sum_{r=1}^{k} {g(i_r)}\cdot g_r$.
Due to the construction, all the entries at the coordinates from $I$ in $g'$ are zero.
If there is a non-zero coordinate left in $g'$, then we need to normalize $g'$, and it to $G$, and update the previous vectors.
So we take the smallest coordinate $i'$ whose entry is non-zero in $g'$, normalize $g'$ to have $1$ in the $i'$-th entry, and add $g'$ to $G$.
To update the previous vectors, for each $r$, $1 \le r \le k$, we set $g_r \gets g_r - g_r(i')\cdot g'$, which results in that $g_r$ has now zero in the $i'$-th entry, and finally we add $i'$ to $I$.
In the opposite case, if $g'=0$, then $g$ belongs to $\lspan(G)$ and thus should not be added.

Note that at any point, the set $G$ is a basis of the linear span of all the processed vectors, which is a straightforward corollary from using the Gaussian approach.

\begin{algorithm}[htb]
\caption{Resizing a subset.}\label{alg:resize_a_subset}
\begin{algorithmic}[1]
\Require An automaton $\mathrsfs{A}=(Q,\Sigma,\delta)$, a~subset $S \subseteq Q$
\State $W_0 \gets \{[Q]\}$
\For{$i$ \textbf{from} $1$ \textbf{to} $n-1$}
    \State $D \gets \{ g[a] \mid g \in W_{i-1}, a\in \Sigma \}$
    \State $W_i \gets \{\}$
    \ForAll{$z \in D$}
        \If{$[S] \scalar z \neq |S|$}
            \State \Return $i$
        \ElsIf{$\Call{Feed}{z}$}
            \State Add $z$ to $W_{i}$
        \EndIf
    \EndFor
    \If{$W_i = \emptyset$}
        \State \Return \textit{None}
    \EndIf
\EndFor
\State \Return \textit{None}
\end{algorithmic}
\end{algorithm}

We now turn to the main procedure of our algorithm, which is shown in Algorithm~\ref{alg:resize_a_subset}.
Our goal is to find the minimum length of a word $w$ such that $|S\cdot w^{-1}| \neq |S|$.
This is equivalent to $[S] \scalar [Q][w] \neq |S|$.
We do this by using a \emph{wave approach} as in breadth-first search.
We start by feeding $[Q]$ to $\Call{Filter}{}$ and let $W_0 = \{[Q]\}$.
Then in each iteration $1 \le i \le n-1$, we consider the set of vectors $D = \{ g[a] \mid g \in W_{i-1}, a\in \Sigma \}$ and build a new subset of independent vectors $W_i$ as follows.
For each vector $z$ from $D$, we first check whether $[S] \scalar z = |S|$.
If this is not the case, we claim that $i$ is the length of a shortest word which changes the size of the preimage of $S$.
Otherwise, we feed $z$ to $\Call{Filter}{}$ and add it to (initially empty) $W_i$ if the corresponding basis vector was added to $G$.
Note that the current $G$ after the $i$-th iteration is equal to $\bigcup_{j=0}^i W_i$.
We stop if either $W_i = \emptyset$ or the last $(n-1)$-th iteration ends, which means that there is no resizing word.

\noindent\textit{Correctness}.
To prove the correctness, note that by the construction all vectors from $W_i$ can be written as $[Q][w]$ for some word $w$ of length $i$.
Thus, if we have found a vector $z \in D$ such that $[S] \scalar z \neq |S|$, this means there is a word $w$ of length $i$ such that $$[S] \scalar [Q][w] = [S \cdot w^{-1}] \scalar [Q] = |S \cdot w^{-1}| \neq |S|.$$

It remains to show that if we get to an $i$-th iteration, then there is no word $w$ of length less than $i$ which violates $[S] \scalar [Q][w] = |S|$.
For $r \geq 0$, denote $U_r = \bigcup_{i=0}^{r} W_i$.
We prove by induction that for each word $w$ of length $r < i$, $[Q][w] \in \lspan([Q][U_{r}])$.
For $r=0$ this is trivial.
If $r>0$, then $w=w' a$ for some $a \in \Sigma$ and by induction $[Q][w'] \in \lspan([Q][U_{r-1}])$, that is,
$$[Q][w'] = \sum_{j=0}^{r-1}\sum_{u \in W_j}\lambda_u [Q][u],$$ 
for some values $\lambda_u \in \mathbb{R}$.
It follows that 
$$[Q][w' a] = [Q][w'][a] = \sum_{j=0}^{r-1}\sum_{u \in W_j}\lambda_u [Q][u][a] = g_v + \sum_{u \in W_{r-1}}\lambda_u [Q][u][a],$$ 
where $g_v \in \lspan([Q][U_{r-1}])$.
By the construction, we feed all vectors of the form $[Q][u][a]$ for $u \in W_{r-1}$ and $a \in \Sigma$ to $\Call{Filter}{}$ function.
Since the added vectors to $G$, and so to $W_r$, are a linear basis of the linear span of all the processed vectors, every vector $[Q][u][a]$ belongs to $\lspan([Q][U_r])$, which proves the induction step.

Thus, if we had a word of length $w$ of length less than $i$ with $[S] \scalar [Q][w] \neq |S|$, we would have $[Q][w] = \sum_{u \in U_{i-1}}\lambda_u [Q][u]$ for some $\lambda_u \in \mathbb{R}$.
Now, on the one hand we have 
\begin{equation}\label{eq:lambda_sum}
n = [Q][w] \scalar [Q] = \sum_{u \in U_{i-1}}\lambda_u ([Q][u] \scalar [Q]) = n\sum_{u \in U_{i-1}}\lambda_u,
\end{equation}
while on the other hand we have  
$$|S| \neq [Q][w] \scalar [S] = \sum_{u \in U_{i-1}}\lambda_u [Q][u] \scalar [S] = \sum_{u \in U_{i-1}}\lambda_u |S|$$
contradicting (\ref{eq:lambda_sum}).

On the other hand, if $W_i$ is empty for an $i<n$, this means that $\lspan([Q][\Sigma^{\leq i}]) = \lspan([Q][\Sigma^{\leq i-1}])$ and by the linear extending argument we know that the same holds for all $j \geq i$, hence there cannot be a word that violates $[S] \scalar [Q][w] = |S|$.
Note that if there is no resizing word, then we always have this case for some $i<n$, because $\dim(\lspan([Q][w]\mid w \in \Sigma^*)) \le n-1$ and the vectors from all $W_j$ are a basis.

We also conclude that $i$ cannot exceed $n-1$, which proves that the shortest resizing words have length at most $n-1$.
Note that the upper bound $n-1$ is the best possible, at least in the cases $|S| \in \{1,n-1\}$, which can be observed in the \v{C}ern\'{y} automata (see Fig.~\ref{fig:cerny4} with $S=\{3\}$).

\noindent\textit{Complexity}.
Assume that in our computational model every arithmetic operation has a unitary cost.
Then clearly a $k$-th call of $\Call{Feed}{}$ can be performed in $\O(kn)$-time.
However, note that, if an exact computation is performed using rational numbers, then we may require to handle values of exponential order, and the total complexity would depend on the algorithms used for particular arithmetic operations.

Notice that at an $i$-th iteration, we call $\Call{Feed}{}$ at most $|\Sigma| |W_i|$ times, since, by the construction, sets $W_i$ are disjoint because the corresponding vectors are independent.
Since the complexity of $\Call{Feed}{}$ is in $\O(n^2)$, all calls work in $\O(|\Sigma|n^3)$-time.
The other operations took amortized time at most $\O(|\Sigma| n^2)$, which is the cost of computing sets $D$ (at most $n$ vectors in sets $W_i$; note that one $g[a]$ can be computed in $\O(n)$ time, because the automaton is deterministic).
Thus, the whole algorithm works in $\O(|\Sigma|n^3)$ time.

The space complexity is at most $\O(|\Sigma|n+n^2)$, which is caused by storing the automaton and at most $\O(n^2)$ vectors in the sets $W_i$, $G$, and $I$.
\end{proof}

The running time $\O(|\Sigma|n^3)$ of the algorithm is quite large (and may require large arithmetic as discussed in the proof), and it is an~interesting open question whether there is a~faster algorithm for Problems~\ref{pbm:resize} and~\ref{pbm:resize_len}.

We note that Problem~\ref{pbm:resize} becomes trivial when the automaton is synchronizing:
A~word resizing the subset exists if and only if $S \ne \emptyset$ and $S \ne Q$, because if $w$ is a~reset word and $\{q\} = Q\cdot w$, then $S\cdot w^{-1}$ is either $Q$ when $q \in S$ or $\emptyset$ when $q \notin S$.
This implies that there exists a~faster algorithm in the sense of expected running time when the automaton over at~least a~binary alphabet is drawn uniformly at random:

\begin{remark}
The algorithm from~\cite{Berlinkov2016OnTheProbabilityToBeSynchronizable} checks in expected $\O(n)$ time (regardless of the alphabet size, which is not fixed) whether a~random automaton is synchronizing, and it is synchronizing with probability $1-\varTheta(1/n^{0.5|\Sigma|})$ (for $|\Sigma|\ge 2$).
Then only if it is not synchronizing we have to use the algorithm from Theorem~\ref{thm:resize_algorithm}.
Thus, Problem~\ref{pbm:resize_len} can be solved for a~random automaton in the expected time
\begin{equation*}
\O(|\Sigma|n^3)\cdot\varTheta(1/n^{0.5|\Sigma|}) + \O(n) = \O(|\Sigma|n^{3-0.5|\Sigma|}) \leq \O(n^2).
\end{equation*}
Note that the bound is independent on the alphabet size, and this is because a~random automaton with a~growing alphabet is more likely to be synchronizing, so less likely we need to use Theorem~\ref{thm:resize_algorithm}.
\end{remark}

%%%%%%%%%%%%%%%%%%%%%%%%%%%%%%%%%%%%%%%%%%%%%%%%%%%%%%%%%%%%
\section{Conclusions}

We have established the computational complexity of problems related to extending words.
Indirectly, our results about the complexity imply also the bounds on the length of the shortest compressing/extending words, which are of separate interest.
In particular, PSPACE-hardness implies that the shortest words can be exponentially long in this case, and polynomial deterministic or nondeterministic algorithms in our proofs imply polynomial upper bounds.
For example, the question about the length of the shortest totally extending words (in the equivalent terms of compressing $Q$ to a~subset included in $S$) was recently considered \cite{GonzeJungers2018OnCompletelyReachable}, and from our results (PSPACE-completeness) we could infer an~answer that the tight upper bound is exponential.
The algorithm from Theorem~\ref{thm:totallyextensible_large} implies also a bound on the length of the shortest avoiding words for a subset.
That length is at least cubic, which is useless in the case of synchronizing automata, since reset words can be used as avoiding words and there exists a cubic upper bound on the length of the shortest reset words \cite{Shitov2019,Szykula2018ImprovingTheUpperBound}.

Some problems are left open.
In Tables~\ref{tab:complexities_existence} and~\ref{tab:complexities_length} there is a gap.
The complexity of the existence of an extending word when the subset is large (Problem~\ref{pbm:extensible_large}) and the automaton is strongly connected is unknown.
The same holds in the case when the length of the extending word is bounded (Problem~\ref{pbm:extensible_large_len}); now, we can only conclude that it is NP-hard, which follows from Corollary~\ref{cor:totallyextending_large_len_sych}.
The proof of Theorem~\ref{thm:extensible_large} relies on the automaton being not strongly connected.

Further questions may concern other complexity classes like NL (cf.\ Theorem~\ref{thm:totallyextensible_synchro}).
Also, one could try improving the complexity of algorithms, in particular, those from Theorems~\ref{thm:avoiding_characterization} and~\ref{thm:totallyextensible_large} for avoiding words, and also that from Theorem~\ref{thm:resize_algorithm} for resizing words.

%%%%%%%%%%%%%%%%%%%%%%%%%%%%%%%%%%%%%%%%%%%%%%%%%%%%%%%%%%%%
\section*{Acknowledgements}

We thank the anonymous referee for careful reading and detailed comments.
This work was supported by the Competitiveness Enhancement Program of Ural Federal University (Mikhail Berlinkov), and by the National Science Centre, Poland under project number 2014/15/B/ST6/00615 (Robert Ferens) and 2017/25/B/ST6/01920 (Marek Szyku{\l}a).

\bibliographystyle{plainurl}
\bibliography{synchronization}
\end{document}